\documentclass[reqno,12pt]{amsart}
\usepackage{fullpage}
\usepackage{amsfonts}
\usepackage{amssymb}
\usepackage{tikz}

\numberwithin{equation}{section}
\def\Lagr{\mathcal{L}}
\def\Ham{\mathcal{H}}
\def\solnsp{{\mathcal E}}
\def\Dt{\mathcal{D}_t}

\def\dddot#1{\overset{...}{#1}}

\def\X{{\bf X}}
\def\Y{{\bf Y}}
\def\pr{{\rm pr}}

\def\com/{constant of motion}
\def\coms/{constants of motion}
\def\eom/{equations of motion}
\def\LRL/{Laplace-Runge-Lenz}

\def\Rnum{\mathbb{R}}

\def\t{{\rm t}}
\def\det{\mathrm{det}}

\def\bigint{\displaystyle\int}
\def\parder#1#2{\frac{\partial{#1}}{\partial{#2}}}

\def\parders#1#2#3{\frac{\partial^2{#1}}{\partial{#2}\partial{#3}}}
\def\totder#1#2{\frac{d{#1}}{d{#2}}}

\newtheorem{prop}{Proposition}
\newtheorem{thm}{Theorem}
\newtheorem{cor}{Corollary}
\newtheorem{lem}{Lemma}
\newtheorem{defn}{Definition}

\def\Ref#1{Ref.\cite{#1}}

\begin{document}

\title{A hybrid Lagrangian-Hamiltonian framework\\ and its application to conserved integrals\\ and symmetry groups}

\author{
Stephen C. Anco
\\\lowercase{\scshape{
Department of Mathematics and Statistics, 
Brock University\\
St. Catharines, ON, Canada}} \\
}

\begin{abstract}
A hybrid framework is developed that highlights and unifies
the most important aspects of the Noether correspondence between
symmetries and conserved integrals in Lagrangian and Hamiltonian mechanics.
Several main results are shown:
(1) a modern form of Noether's theorem is presented that
uses only the equations of motion, with no knowledge required of an explicit Lagrangian;
(2) the Poisson bracket is formulated with Lagrangian variables
and used to express the action of symmetries on conserved integrals;
(3) features of point symmetries versus dynamical symmetries are clarified and explained;
(4) both autonomous and non-autonomous systems are treated on an equal footing.
These results are applied to dynamical systems that are locally Liouville integrable.
In particular, they allow finding the complete Noether symmetry group of such systems. 
\end{abstract}

\maketitle

\section{Introduction}\label{sec:intro}
\label{intro}

Classical mechanics is a beautiful mathematical area of study
which has several formulations --- Lagrangian, Hamiltonian, Hamilton-Jacobi ---
each with its own viewpoint.
The main unifying thread that relates them is the important correspondence
between conserved integrals (invariants) and symmetries.
Some aspects of this correspondence are best understood from the Lagrangian viewpoint,
such as variational symmetries and Noether's theorem. 
Other aspects are best understood within the Hamiltonian approach,
such as the homomorphism relating the Lie algebra of variational symmetries
and the Poisson bracket algebra of conserved integrals. 
This motivates developing a hybrid approach that brings together
the most advantageous parts of the Lagrangian and Hamilton frameworks. 

The present work develops such a hybrid framework
with the aim of highlighting the Noether correspondence
between conserved integrals and symmetries.
In particular, this framework has several major objectives:
\begin{itemize}
\item
  provide a modern form of Noether's theorem that uses only the equations of motion 
  with no knowledge of an explicit Lagrangian being needed;
\item
  import the Poisson bracket and its usage into a formalism based on Lagrangian variables; 
\item
  express the action of symmetries on conserved integrals through the Poisson bracket;
\item
  clarify and explain the features of point symmetries versus dynamical symmetries;
\item
  treat both autonomous and non-autonomous systems on the same footing. 
\end{itemize}

As a general perspective guiding the developments,
it is essential to recognize that 
any well-posed dynamical system with $N$ degrees of freedom 
necessarily possesses $2N$ independent locally conserved integrals \cite{WhiWat-book,Arn-book,BA-book}.
Local conservation means that the integrals may be only piecewise continuous
when evaluated on solution trajectories of the equations of motion. 
This contrasts with global conservation,
where all integrals are continuous on every solution trajectory,
which is often taken as a requirement in modern Hamiltonian frameworks.
In particular,
when a dynamical system possesses exactly $N$ globally conserved integrals,
it is called completely integrable,
while existence of more than $N$ globally conserved integrals
is a restrictive condition, called superintegrability \cite{Mil.Pos.Win}. 

However, there are many physical systems in which global conservation fails to hold
for some conserved integrals.
A primary example is the \LRL/ vector in central force dynamics
\cite{Fra,Per,GolPooSaf}, 
which is globally conserved for the Kepler problem and for the isotropic oscillator,
yet otherwise is only locally conserved. 
Such conserved integrals nevertheless can have an important mathematical and physical utility \cite{AncMeaPas},
despite being only piecewise continuous on solution trajectories.

When a dynamical system possesses a Lagrangian or a Hamiltonian,
its conserved integrals correspond to infinitesimal symmetries
that preserve the variational structure.
Specifically, Noether's theorem show that this correspondence is one-to-one.
Integrals that are only conserved locally correspond to
variational symmetries that are dynamical, 
but the converse is not true generally.
More generally, detection of global conservation of a conserved integral is
unrelated to existence of (special types of) symmetries \cite{AncBalGan},
since symmetry analysis deals solely with local properties of the equations of motion. 

Each variational symmetry,
whether it is a point symmetry or a dynamical symmetry, 
generates a one-dimensional Lie group of symmetry transformations,
which act on the Lagrangian variables or equivalently on the phase space variables
of a dynamical system. 
These transformation groups comprise \cite{Olv-book,BA-book} 
point transformations,
in which the compositions of transformations hold independently of equations of motion,
and dynamical transformations,
whose compositions close only on solution trajectories. 
There is an equivalent formulation of transformation groups
using an extended space of variables given by adjoining the time variable. 
This formulation has a natural gauge freedom,
which is apparent from a Lagrangian viewpoint, 
but is not usually considered in a Hamiltonian framework. 
However, the gauge freedom will be shown to be highly useful
in obtaining an explicit form for the symmetry transformation group
generated by an infinitesimal dynamical symmetry. 

Infinitesimal variational symmetries have two main properties,
well known in the Lagrangian framework \cite{Olv-book,BA-book},
that as vector fields they are closed under commutation
and act as mappings from the set of conserved integrals into itself. 
As main theorems,
the symmetry action will be expressed in terms of the Poisson bracket,
and the commutator of two variational symmetries will be expressed
by the Noether correspondence through the Poisson bracket.

An important class of dynamical systems consists of
Liouville integrable systems defined by the features that 
the conserved integrals are globally continuous on all solution trajectories, 
and the number of Poisson commuting conserved integrals is maximal,
namely equal to the number of degrees of freedom.
These systems can be integrated by the introduction of action-angle variables
associated to the commuting conserved integrals. 
As an application of the main theorems,
several interesting results will be derived
for systems that are Liouville integrable in a local sense. 
First,
local Liouville integrability will be formulated in terms of variational symmetries.
Second,
the action-angle variables will be shown to yield additional locally conserved integrals.
This will provide a complete description of the Noether symmetry group
for these systems. 

The rest of the paper is organized as follows.
Section~\ref{sec:cons.integrals} reviews locally conserved integrals
for general dynamical systems. 
Section~\ref{sec:noether} explains the modern statement of Noether's theorem
and related developments. 
Section~\ref{sec:symmgroup} considers infinitesimal dynamical symmetries
and the corresponding Lie groups of symmetry transformations
acting on the Lagrangian and Hamilton variables. 
The equivalent formulation of these transformations
acting on the extended space of variables is explained,
with emphasis on its inherent gauge freedom. 
Section~\ref{sec:PB.symmaction} presents the connection between
Poisson brackets and symmetry actions on conserved integrals
in a Lagrangian framework. 
Section~\ref{sec:liouville.system} applies the preceding results to dynamical systems that are locally Liouville integrable.
Some concluding remarks are made in Section~\ref{sec:conclude}.

An appendix contains a proof of the result connecting Poisson brackets and symmetry actions. 
Throughout the paper,
the summation convection will be used for repeated Latin indices.

\section{Conserved integrals of dynamical systems}\label{sec:cons.integrals}

Consider a general dynamical system with variables $q^i$, $i=1,\ldots,N$, 
whose equations of motion
\begin{equation}\label{eom.gen}
  \ddot{q}^ i = f^i(t,q,\dot{q})
\end{equation}
arise from a Lagrangian $\Lagr(t,q,\dot{q})$,
\begin{equation}\label{EL.eqn.gen}
  \frac{\delta\Lagr}{\delta q^i}
  = \parder{\Lagr}{q^i} -\frac{d}{dt} \parder{\Lagr}{\dot{q}^i}
  = 0 . 
\end{equation}  
This holds when the right-hand side of the equations of motion \eqref{eom.gen}
satisfies
\begin{equation}\label{f.eom}
  f^i = g^{-1}{}^{ij}\Big( \parder{\Lagr}{q^j} - \parders{\Lagr}{t}{\dot{q}^j} -\dot{q}^k \parders{\Lagr}{q^k}{\dot{q}^j} \Big)
\end{equation}
where $g^{-1}{}^{ij}$ is the inverse of the Hessian matrix
\begin{equation}\label{g.matrix}
  g_{ij} =  g_{ji} = \parders{\Lagr}{\dot{q}^i}{\dot{q}^j} . 
\end{equation}
Specifically,
\begin{equation}\label{EL.id.gen}
  \frac{\delta\Lagr}{\delta q^i} = g_{ij} (f^j - \ddot{q}^j)
\end{equation}
where it is assumed that the Lagrangian is non-degenerate, namely $\det(g) \neq 0$.
A Lagrangian is unique up to the addition of a total time-derivative,
which changes the action principle
\begin{equation}\label{action.gen}
  S = \int_{t_1}^{t_2} \Lagr(t,q(t),\dot q(t))\, dt
\end{equation}
only by an irrelevant endpoint term.
In particular,
under $\Lagr \to \Lagr + \dot{\mathcal{A}}$ for any function $\mathcal{A}(t,q)$,
both $g_{ij}$ and $f^i$ are unchanged. 

A conserved quantity, or first integral, is a function $C(t,q,\dot{q})$
such that its time derivative vanishes for all solutions $q^i(t)$ of the equations of motion,
\begin{equation}\label{conserved}
  \dfrac{d}{dt}C(t,q(t),\dot{q}(t)) =0,
  \quad
  \ddot{q}(t) = f^i(t,q(t),\dot{q})  . 
\end{equation}
When conservation \eqref{conserved} holds for all times $t$,
the quantity $C$ is globally conserved.
This condition is often taken to be part of the definition of conserved quantities
in Hamiltonian mechanics, especially when global properties of solution trajectories are of primary interest. 
However, a quantity $C$ may be conserved only locally in time and still provide
useful information about local properties of solution trajectories.

A well-known example is the \LRL/ vector in central force dynamics
\cite{Fra,Per,GolPooSaf}, 
which is globally conserved for the Kepler problem and for the isotropic oscillator,
yet otherwise is only locally conserved \cite{AncMeaPas}. 
In particular, on a given trajectory $q^i(t)$,
suppose that a periapsis point occurs at time $t^*$,
namely where $q^i(t^*)$ is a local minimum of $|q^i(t)|$.
At any time $t<t^*$, the LRL vector points in the direction of this periapsis point,
and when $q^i(t)$ reaches this point in the trajectory,
then the LRL vector jumps to the next periapsis point 
if the trajectory describes a precessing orbit having multiple periapsis points;
otherwise if the trajectory describes an orbit with a single periapsis,
the LRL vector remains constant for all times $t$.
Consequently, in general central force dynamics,
the LRL vector is locally conserved and globally discontinuous
at periapsis points on solutions trajectories that exhibit precession. 

Since one objective of the present work is to focus on local properties of dynamical systems, 
a precise formulation of local conservation will be given.

\begin{defn}\label{defn:conserved}
A function $C(t,q,\dot{q})$ is a \emph{locally conserved integral}
if its conservation \eqref{conserved} holds piecewise in $t$
for every solution $q^i(t)$ of the equations of motion.
If for all solutions the conserved integral is continuous for all times $t$,
then it is \emph{globally conserved}. 
\end{defn}

Conserved quantities can be usefully classified into two different types
depending on whether or not $t$ appears explicitly in the function $C$. 

\begin{defn}
If a conserved quantity does not depend explicitly on $t$,
then it is called a constant of motion,
$C(q,\dot{q})$.
A conserved quantity $C(t,q,\dot{q})$, with $C_t\neq 0$, is otherwise
called an integral of motion.
\end{defn} 

The dynamical systems considered here will be assumed to be locally well-posed,
meaning that, for any initial conditions $(q^i(t_0),\dot{q}^i(t_0))$ posed at $t=t_0$,
there exists a unique solution $q^i(t)$ locally in time $t\geq t_0$.
For such systems, 
this existence and uniqueness implies the following result. 

\begin{thm}\label{thm:conserved.integrals}
A well-posed dynamical system with $N$ degrees of freedom
possesses $2N$ locally conserved integrals
which are functionally independent.
If the system is autonomous, then
these conserved integrals comprise
$2N-1$ local constants of motion plus one local integral of motion. 
\end{thm}

\begin{proof}
Consider the initial-value solution $q^i = \phi^i(t)$
and its derivative $\dot{q}^i = \phi^i(t)'$,
which constitute a system of $2N$ equations
in terms of the initial values $(q^i(t_0),\dot{q}^i(t_0))$. 
By the implicit function theorem,
the system has a local solution
$(q^i(t_0),\dot{q}^i(t_0)) = (C_i(t,q,\dot{q}),C_{N+i}(t,q,\dot{q}))$.
Each of these expressions is a function $C_j(t,q,\dot{q})$ that satisfies $\dot{C_j}=0$, 
namely a locally conserved integral.
Since the initial values can be chosen freely,
this implies that the $2N$ conserved integrals
$C_1,\dots,C_{2N}$ are functionally independent.

To obtain $2N-1$ constants of motion from this set,
suppose those conserved integrals that are integrals of motion
comprise a subset, say $C_1(t,q,\dot{q}),\ldots,C_n(t,q,\dot{q})$, with $n\geq 2$,
and let $A_1=\partial_tC_1,\ldots,A_n=\partial_tC_n$ which are non-zero. 
Substitute $q^i = \phi^i(t)$ and $\dot{q}^i = \phi^i(t)'$ into each $A_j$,
with the initial values replaced by $q^i(t_0)= C_i$ and $\dot{q}^i=C_{N+i}$.
If the dynamical system is autonomous,
then $\dot{A}_j = (\partial_t C)\dot{} = \partial_t \dot{C}_j =0$,
and thus $A_j$ is a function only of $C_1,\ldots,C_{2N}$. 

Now consider the linear homogeneous partial differential equation
$0= \partial_t F(C_1,\ldots,C_n)
= A_1\partial_{C_1}F + \cdots + A_n\partial_{C_n}F$.
It will admit $n-1$ particular solutions $F_1,\ldots,F_{n-1}$ 
which are functions of $C_1,\ldots,C_{2N}$.
This set represents $n-1$ functionally independent constants of motion,
whereby $F_1,\ldots,F_{n-1}$ together with $C_{n+1},\ldots,C_{2N}$
constitute $2N-1$ constants of motion.
The integral of motion can be chosen as any one of $C_1,\ldots,C_n$.
\end{proof}

Any physically reasonable dynamical system obeys local well-posedness,
and thus it necessarily admits $2N$ locally conserved integrals.
Of course, this existence argument does not provide a way to find
explicit expression for these integrals in the absence of knowing
the general initial-value solution explicitly. 
It is worth emphasizing that the existence of these $2N$ local integrals
has no implications about global integrability of the system.

\section{Noether's theorem in modern form}\label{sec:noether}

Noether's theorem is commonly formulated as connecting
continuous symmetry groups to conserved quantities.
A symmetry group in this context is a Lie group of transformations
on the variables $(t,q^i,\dot{q}^i)$ such that the action principle \eqref{action.gen} is invariant.
In the case of single dynamical variable, $N=1$,
the most general continuous transformation acting in the coordinate space $(t,q,\dot{q})$
consists of a \emph{contact transformation} \cite{BCA-book}
\[ 
(t,q,\dot{q})\to (t^\dagger,q^\dagger,\dot{q}^\dagger)  =
(t,q,\dot{q})  + \varepsilon \big( \tau(t,q,\dot{q}),\eta(t,q,\dot{q}), \sigma(t,q,\dot{q}) \big)+ O(\varepsilon^2)
\]
with parameter $\varepsilon$,
under which the contact condition $dq = \dot{q}dt$ is required to be preserved.
This condition can be shown to determine 
$\xi = P_q$, $\eta = qP_q-P$, $\sigma = -P_t - \dot{q}P_q$
in terms of a freely specified function $P(t,q,\dot{q})$. 
In contrast, for the case of more than one dynamical variable, $N\geq2$,
the counterpart of a contact transformation is a prolonged \emph{point transformation}
\[
(t,q^i,\dot{q}^i)\to (t^\dagger,q^i{}^\dagger,\dot{q}^i{}^\dagger)  =
(t,q^i,\dot{q}^i)  + \varepsilon \big( \tau(t,q),\eta^i(t,q), \eta_{(1)}^i(t,q,\dot{q}) \big)+ O(\varepsilon^2)
\]
where $\eta_{(1)}^i = \dot\eta^i -\dot{\tau}\dot{q}^i$ is determined by
$dq^i{}^\dagger = \dot{q}^\dagger dt^\dagger$
in terms of two freely specified functions $\tau(t,q)$ and $\eta(t,q)$. 
Thus, a priori, single variable dynamical systems can have a larger symmetry group
in comparison to multi-variable dynamical systems.

However, the conserved quantity that arises from a continuous symmetry group
turns out to involve only the infinitesimal generator of the transformations. 
As a consequence, Noether's theorem actually has a more general formulation
utilizing infinitesimal generators whose form depends on a function of
$t$, $q^i$, $\dot{q}^i$.
Such generators were, in fact, considering in Noether's original work
\cite{Noe} (see also \Ref{Olv-book}). 
The corresponding finite transformations,
which act on the infinite space of variables $(t,q,\dot{q},\ddot{q},\ldots)$ 
containing time derivatives of all orders,
will not be needed. 

To state the general formulation succinctly, 
the appropriate mathematical setting will be the finite jet space, $J$,
coordinatized by $(t,q^i,\dot{q}^i,\ddot{q}^i)$.
This space has a natural fibering 
given by the position, velocity, and acceleration variables 
$(q^i,\dot{q}^i,\ddot{q}^i)\in \Rnum^{3N}$, 
at each time $t \in \Rnum$, 
where the fiber is called the vertical space.
The following main geometrical objects will be utilized: 
vertical vector fields of the form 
$\X_P = P^i(t,q,\dot{q}) \partial_{q^i}$; 
the Lagrangian 1-form $\Lagr(q,\dot{q})\,dt$;
the total time-derivative vector field
$D_t = \partial_t + \dot{q}^i\partial_{q^i} + \ddot{q}^i\partial_{\dot{q}^i}$;
contact relations 
$dq^i = \dot{q}^i\,dt$ and $d\dot{q}^i = \ddot{q}^i\,dt$.

One useful identity, holding due to the contact relations, is
$dF = D_t F\, dt$ for any function $F(t,q,\dot{q})$.
A second useful identity is that $d$ commutes with $\pr\X_P$,
due to $[\pr\X_P,D_t] = 0$,
where $\pr\X_P = P^i\partial_{q^i} + \dot{P}^i\partial_{\dot{q}^i}$
denotes the prolongation of $\X$ as a vertical vector field.

\subsection{Variational symmetries}

In the infinitesimal formulation of variational symmetries
for a Lagrangian dynamical system \eqref{EL.eqn.gen} with $N\geq 1$ variables $q^i$, 
the number of functions that are freely specifiable a priori
is the same as for infinitesimal contact transformations. 

\begin{defn}\label{defn:varsymm}
An \emph{infinitesimal variational symmetry} is a vertical vector field 
\begin{equation}\label{X.gen}
  \X_P = P^i(t,q,\dot{q}) \partial_{q^i}
\end{equation}
such that the Lagrangian 1-form is invariant up to an exact 1-form
\begin{equation}\label{X.varsymm.cond}
  \pr\X_P\rfloor d(\Lagr\,dt) = dW
\end{equation}
for some function $W(t,q,\dot{q})$,
under the prolonged vector field
\begin{equation}\label{pr.X.gen}
  \pr\X_P = P^i\partial_{q^i} + \dot{P}^i\partial_{\dot{q}^i} . 
\end{equation}
\end{defn}

The condition of symmetry invariance \eqref{X.varsymm.cond}
holds modulo the contact relations.
In particular,
\[
\pr\X_P\rfloor d(\Lagr\,dt)
= \big( \pr\X_P\rfloor d\Lagr(q,\dot{q}) \big)\, dt
=  \big( P^i \partial_{q^i}\Lagr + \dot{P}^i \partial_{\dot{q}^i}\Lagr \big)\,dt
\]
and
\[
dW = \big( \partial_t W +\dot{q}^i \partial_{q^i}\Lagr + \ddot{q}^i \partial_{\dot{q}^i}\Lagr \big)\,dt = D_t W\,dt , 
\]
which shows that 
\begin{equation}\label{inv.cond.gen}
  \pr\X_P(\Lagr) 
  =  \totder{W}{t} . 
\end{equation}
Thus, the action principle \eqref{action.gen} is preserved up to an irrelevant end point term. 
This is equivalent to the condition 
\begin{equation}\label{inv.cond.alt.gen}
 \frac{\delta \pr\X_P(\Lagr)}{\delta q^i} =0
\end{equation}
since the variational derivative annihilates a function
if and only if the function is a total time derivative.
The latter statement has a simple formulation in the present setting: 
$d^2F=0$ is necessary and sufficient for $F=dW$.
This reflects the fact that the jet space $J$ has trivial cohomology,
which holds as a consequence of a general result in the variational bi-complex
\cite{Olv-book}. 

\begin{prop}\label{prop:symm.cond.withoutW.gen}
A vertical vector field \eqref{X.gen} is an infinitesimal variational symmetry
if and only if it satisfies 
\begin{equation}\label{inv.cond.alt2.gen}
  d\big(\pr\X_P\rfloor d(\Lagr\,dt)\big)= 0 . 
\end{equation}
\end{prop}

Therefore, knowledge of the function $W(t,q,\dot{q})$ is unnecessary
for determining when a vector vector field is a symmetry.
Moreover, the invariance condition \eqref{inv.cond.alt2.gen} is unchanged
if the Lagrangian is altered by the addition of any total time derivative,
$\Lagr \to \Lagr + D_t\mathcal{A}$,
since $d(\Lagr\,dt) \to d(\Lagr\,dt +d\mathcal{A}) = d(\Lagr\,dt)$
for any function $\mathcal{A}(t,q,\dot{q})$. 

There is a further useful rewriting of the invariance condition \eqref{inv.cond.alt2.gen}
in terms of the equations of motion \eqref{eom.gen}. 
For any function $F(t,q,\dot{q},\ddot{q})$,
introduce the Euler-Lagrange operators $E_i^{(0)}$, $E_i^{(1)}$, $E_i^{(2)}$ 
via the identity
\begin{equation}\label{EL.ops.id}
  dF = \partial_t F\,dt + E_i^{(0)}(F)\,dq^i + D_t(E_i^{(1)}(F)\,dq^i) + D_t^2(E_i^{(2)}(F)\,dq^i)
\end{equation}
which is derived through integration by parts, holding in the twice-prolonged jet space. 
Explicit expressions for these operators are given by
\begin{equation}\label{higher.EL.ops}
  E_i^{(0)} = \partial_{q^i} - D_t \partial_{\dot{q}^i} + D_t^2 \partial_{\ddot{q}^i} , 
  \quad
  E_i^{(1)} = \partial_{\dot{q}^i} - 2D_t \partial_{\ddot{q}^i} ,
  \quad
  E_i^{(2)} = \partial_{\ddot{q}^i} . 
\end{equation}
They can be shown to satisfy the following properties:
$E_i^{(0)}(F) =0$ 
if and only if $F=D_t W$ for some function $W(t,q,\dot{q})$,
since $E_i^{(0)}$ is simply the variational derivative $\dfrac{\delta}{\delta q^i}$;
similarly, 
$E_i^{(1)}(F) =0$
if and only if $F=D_t^2 W$ for some function $W(t,q)$.
(Note that everything here has a natural extension to any higher-order jet space
\cite{Olv-book,Anc-review}.)

\begin{lem}\label{lem:noether.id}
Any vertical vector field \eqref{X.gen}
satisfies the variational (Noether) identity 
\begin{equation}\label{noether.id}
  \pr\X_P\rfloor d(\Lagr\,dt)
  = P^i E_i^{(0)}(\Lagr)\,dt + d (P^i E_i^{(1)}(\Lagr)) 
\end{equation}
where $E_i^{(0)}(\Lagr) = g_{ij} (f^j - \ddot{q}^j)$ is proportional to the equations of motion \eqref{eom.gen}. 
\end{lem}

The proof amounts to the following easy computation: 
\[\begin{aligned}
\pr\X_P\rfloor d \Lagr
& =\pr\X_P\rfloor \big( E_i^{(0)}(F)\,dq^i + D_t(E_i^{(1)}(F)\,dq^i) \big) \\
& = P^i E_i^{(0)}(F)  + D_t( P^i E_i^{(1)}(F) ) 
\end{aligned}\]
since vector fields commute with total derivatives.
Multiplying by $dt$ and expressing $D_t$ as a total differential 
then yields the identity \eqref{noether.id}.

Now, combining the variational identity \eqref{noether.id}
and the left-hand side of the symmetry invariance condition \eqref{inv.cond.alt2.gen}
directly gives 
\begin{equation}
  d\big( \pr\X(P)\rfloor d(\Lagr\,dt) \big)
  = d\big( P^i E_i^{(0)}(\Lagr)\,dt + d (P^i E_i^{(1)}(\Lagr)) \big)
  = d\big( P^i E_i^{(0)}(\Lagr)\,dt \big) . 
\end{equation}
Since $E_i^{(0)}(\Lagr) = \dfrac{\delta \Lagr}{\delta q^i}$ is proportional to the equations of motion \eqref{EL.id.gen},
condition \eqref{inv.cond.alt2.gen} becomes
\begin{equation}
 d\big( P^i g_{ij}(f^j - \ddot{q}^j) \,dt \big) = 0 . 
\end{equation}
This establishes the following modern statement of symmetry invariance,
which is a counterpart of a similar result for systems of partial differential equations
\cite{Anc-review}. 

\begin{prop}\label{prop:symm.cond.EL.gen}
A vertical vector field \eqref{X.gen} is an infinitesimal variational symmetry
if and only if it satisfies
\begin{equation}\label{inv.cond.ELeqns.gen}
 d\big( \X_P\rfloor (g_{ij}(\ddot{q}^j -f^j) dq^i) \,dt \big) = 0 , 
\end{equation}
which involves only the equations of motion \eqref{eom.gen}
and the Hessian matrix \eqref{g.matrix}. 
\end{prop}

\subsection{Multipliers}

Turning now to Noether's theorem,
a preliminary step will be to formulate the condition 
for a function $C(t,q,\dot{q})$ to be a locally conserved quantity.
Let $\solnsp$ denote the space of solutions of the equations of motion \eqref{eom.gen}. 
This space is formally given by the set of surfaces 
$\ddot{q}^i - f^i =0$, $i=1,\ldots,N$, in the jet space $J$.

Any time derivative $D_tC$ can be expanded by chain rule to get
\begin{equation}\label{DtC}
  \dot{C} = \Dt C +(\ddot{q}^i -f^i) \partial_{\dot{q}^i} C
\end{equation}
where
\begin{equation}\label{time.der.solnsp}
\Dt
= \partial_t + \dot{q}^i\partial_{q^i} + f^i\partial_{\dot{q}^i}
\end{equation}
represents the time derivative restricted to the solution space.
This immediately yields the requisite necessary and sufficient condition. 

\begin{prop}\label{prop:conserved}
A function $C(t,q,\dot{q})$ is a locally conserved quantity
if and only if it identically satisfies
\begin{equation}\label{local.conserved}
  \Dt C =0
\end{equation}
\end{prop}

As a consequence, off of the solution space,
a locally conserved quantity obeys 
\begin{equation}\label{multr.eqn.gen}
  \dot{C} = (\ddot{q}^i - f^i)  Q_i,
  \quad
  Q_i = \partial_{\dot{q}^i} C
\end{equation}
where the function $Q_i(t,q,\dot{q})$ is called a \emph{multiplier}.
A natural geometric meaning is known for multipliers as 1-forms
$Q_i (d\ddot{q}^i - d f^i)$,
as explained in \Ref{AncBao},
but this meaning will not be needed here. 

The essence of Noether's theorem is that it determines
a direct relationship between multipliers and variational symmetries.

\subsection{Noether correspondence between variational symmetries and conserved integrals}

Return to the variational identity \eqref{noether.id}
and the original form of the symmetry condition \eqref{X.varsymm.cond}.
Together they yield
\begin{equation}\label{noether.rel.ELop}
d(W - P^i E_i^{(1)}(\Lagr))  = P^i E_i^{(0)}(\Lagr)\,dt , 
\end{equation}
which has the explicit component form 
\begin{equation}\label{noether.rel}
D_t\big( P^i \partial_{\dot{q}^i}\Lagr  - W \big) 
=g_{ij} (\ddot{q}^i - f^i) P^j . 
\end{equation}
using the Euler-Lagrange relation \eqref{EL.id.gen}. 
Since the right-hand side vanishes on the solution space $\solnsp$, 
a conserved quantity is obtained, 
\begin{equation}\label{conserved.gen}
  \dot{C}\big|_{\solnsp} =0,
  \quad
  C   = P^i \partial_{\dot{q}^i}\Lagr  - W . 
\end{equation}
This is the most familiar general statement of Noether's theorem:
\emph{every infinitesimal variational symmetry produces a locally conserved quantity}.

Importantly, the converse statement also holds.
For any locally conserved quantity $C(t,q,\dot{q})$,
the multiplier equation \eqref{multr.eqn.gen}
can be expressed as
\begin{equation}
  \dot{C} = (\ddot{q}^i - f^i)  g_{ij} P^i,
  \quad
  P^i = g^{-1}{}^{ij}Q_j
\end{equation}
using the invertibility of the Hessian matrix $g_{ij}$. 
Then the Euler-Lagrange relation \eqref{EL.id.gen} yields
\begin{equation}
  P^i \frac{\delta\Lagr}{\delta q^i} = P^i E_i^{(0)}(\Lagr) = -\dot{C} . 
\end{equation}
Substitution of this expression into the variational identity \eqref{noether.id}
leads to the relation
\begin{equation}
  \pr\X_P\rfloor d(\Lagr\,dt)
  = d (P^i E_i^{(1)}(\Lagr) - C) 
\end{equation}
where $\pr\X_P$ is the prolonged vertical vector field \eqref{pr.X.gen} 
as defined in terms of the function $P^i$. 
This implies that the symmetry condition holds in its original form \eqref{X.varsymm.cond},
with 
\begin{equation}
  W = P^i E_i^{(1)}(\Lagr) - C = P^i \partial_{\dot{q}^i}\Lagr - C . 
\end{equation}
Hence,
\emph{every locally conserved quantity arises from an infinitesimal variational symmetry},
which is the converse of the familiar form of Noether's theorem stated previously. 

These developments show that the full statement of Noether's theorem
provides a one-to-one correspondence between
locally conserved quantities and infinitesimal variational symmetries.
In fact, this correspondence can be stated in a completely explicit form with a little more work,
using the multiplier equation \eqref{multr.eqn.gen}.

\begin{thm}\label{thm:noether}
For a dynamical system with a Lagrangian formulation \eqref{EL.eqn.gen},
a function $C(t,q,\dot{q}^i)$ is a locally conserved quantity \eqref{local.conserved}
if and only if
the vertical vector field $\X =P^i(t,q,\dot{q}) \partial_{q^i}$
is an infinitesimal variational symmetry \eqref{inv.cond.alt2.gen}, 
where $C$ and $P^i$ are explicitly related via the Hessian matrix \eqref{g.matrix} by
\begin{equation}\label{PfromC.gen}
  P^ i = g^{-1}{}^{ij} C_{\dot{q}^j}
\end{equation}
and 
\begin{equation}\label{CfromP.gen}
  \begin{aligned}
  C(t,q,\dot{q})
  = & \int_{\mathcal{C}} \big(
  {- g_{ij}} f^i P^j + \dot{q}^k ( (g_{ij} f^i P^j)_{\dot{q}^k} + (g_{kj} P^j)_t + \dot{q}^i(g_{kj} P^j)_{q^i} ) \big)\,dt
  \\&\qquad
  -\big( (g_{ij} f^i P^j)_{\dot{q}^k} + (g_{kj} P^j)_t + \dot{q}^i(g_{kj} P^j)_{q^i} \big)\, dq^k
  + g_{kj} P^j \, d\dot{q}^k . 
  \end{aligned}
\end{equation}
Here $\mathcal C$ denotes any curve in the coordinate space $(t,q,\dot{q}^i)$, 
starting at an arbitrary point $(t_0,q^i_0,\dot{q}^i_0)$. 
\end{thm}

Here and hereafter, for ease of notation, subscripts are used to partial derivatives. 

\begin{proof}
Suppose $C(t,q,\dot{q})$ is a conserved quantity.
First, expand the left-hand side of the multiplier equation \eqref{multr.eqn.gen}
by the chain rule to get 
\begin{equation}
  C_{t}  + \dot{q}^i C_{q^i} C  + \ddot{q}^i C_{\dot{q}^i}
  = g_{ij} (\ddot{q}^i - f^i) P^j
\end{equation}
where partial derivatives are denoted here and hereafter by subscripts. 
Next, split this equation with respect to $\ddot{q}^i$, which gives the relations
\begin{equation}\label{C.ddotq}
  C_{\dot{q}^i}  = g_{ij} P^j
\end{equation}
along with 
\begin{equation}\label{C.oth}
  C_{t} + \dot{q}^i C_{q^i} 
  = - g_{ij} f^i P^j . 
\end{equation}
Inverting the first relation yields the result \eqref{PfromC.gen}. 

Conversely, suppose $\X_p = P^i(t,q,\dot{q})\partial_{q^i}$
is an infinitesimal variational symmetry.
First, take the derivative of relation \eqref{C.oth} with respect to $\dot{q}^k$
and substitute the first relation \eqref{C.ddotq} to get 
\begin{equation}\label{C.dotq}
  C_{q^k}  = -\big( (g_{ij} f^i P^j)_{\dot{q}^k} + (g_{kj} P^j)_t + \dot{q}^i (g_{kj} P^j)_{q^i} \big) . 
\end{equation}
Next, substitute this expression back into the second relation \eqref{C.oth},
which yields
\begin{equation}\label{C.t}
 C_{t}   = - g_{ij} f^i P^j + \dot{q}^k \big( (g_{ij} f^i P^j)_{\dot{q}^k} + (g_{kj} P^j)_t + \dot{q}^i(g_{kj} P^j)_{q^i} \big) . 
\end{equation}
Finally, apply a line integral to obtain $C$ from 
$dC = C_{t}\,dt + C_{q^i}\,dq^i + C_{\dot{q}^i}\,d\dot{q}^i$
via expressions \eqref{C.ddotq}, \eqref{C.dotq}, \eqref{C.t},
which yields the result \eqref{CfromP.gen}. 
\end{proof}

The following well-known remark is useful to keep in mind \cite{BA-book}. 
There is a straightforward connection between variational symmetries of a Lagrangian 
and symmetries of the equations of motion.
Since an infinitesimal variational symmetry preserves the action principle up to an end point term,
which does not change the equations of motion,
it necessarily preserves the extrema of the action principle
and thus yields an infinitesimal symmetry of the equations of motion.
The converse, however, is not true in general,
as there may exist symmetries of the equations of motion --- 
such as scalings --- that do not preserve the action principle.

The sequel will use the prolonged jet space $J^{(1)}$
with coordinates $(t,q^i,\dot{q}^i,\ddot{q}^i,\dddot{q}{}^i)$. 

\begin{defn}\label{defn:symm}
An \emph{infinitesimal symmetry} of the equations of motion \eqref{eom.gen}
is a vertical vector field
\begin{equation}\label{X.gen.alt}
  \X_P = P^i(t,q,\dot{q}) \partial_{q^i}
\end{equation}
whose prolongation 
\begin{equation}
\pr\X_P = P^i\partial_{q^i} + \dot{P}^i\partial_{\dot{q}^i} + \ddot{P}^i\partial_{\ddot{q}^i}
\end{equation}
leaves invariant the solution space $\mathcal E$: 
\begin{equation}\label{X.symm.cond}
  \pr\X_P(\ddot{q}^i - f^i) \big|_{\solnsp} = 0,
  \quad
  i=1,\ldots,N . 
\end{equation}
\end{defn}

The symmetry condition \eqref{X.symm.cond}
can be expressed more explicitly by means of
the time derivative \eqref{time.der.solnsp} on $\solnsp$.
When the action of a vertical vector field on a function $F(t,q,\dot{q})$
is restricted to $\solnsp$,
it is given by $\pr\X_P(F)|_\solnsp = P^i\partial_{q^i}F + \Dt P^i\partial_{\dot{q}^i}F$
which thereby defines the associated vector field
\begin{equation}\label{X.gen.solnsp}
  \X_P^\solnsp = P^i\partial_{q^i} + \Dt P^i\partial_{\dot{q}^i}
\end{equation}
in the solution space $\solnsp \subset J$.
Note that, here,
both $P^i$ and $\Dt P^i= P^i_t +\dot{q}^j P^i_{q^j} +f^j P^i_{\dot{q}^j}$
are functions only of $t$, $q^i$, and $\dot{q}^i$.

\begin{prop}
A vertical vector field \eqref{X.gen.alt} is an infinitesimal symmetry of the equations of motion
if and only if the associated vector field \eqref{X.gen.solnsp} satisfies 
\begin{equation}\label{X.symm.cond.alt}
  \Dt^2 P^i - \X_P^\solnsp(f^i) =0,
  \quad
  i=1,\ldots,N . 
\end{equation}
\end{prop}

This result can be taken as an alternative definition of
an infinitesimal symmetry formulated in the solution space
directly in terms of the projected vector field \eqref{X.gen.solnsp}. 

When this vector field arises from an infinitesimal variational symmetry \eqref{inv.cond.alt2.gen},
it can be expressed explicitly in terms of the locally conserved quantity coming from 
the Noether correspondence, 
which yields the following corollary of Theorem~\ref{thm:noether}. 

\begin{cor}\label{cor:noether.symm}
For any locally conserved quantity $C(t,q,\dot{q})$,
there is a corresponding infinitesimal symmetry of the equations of motion \eqref{eom.gen}, 
\begin{equation}\label{varsym.solnsp.gen}
  \X^\solnsp_{(C)} = g^{-1}{}^{ij}Q_j \partial_{q^i} + \Dt(g^{-1}{}^{ij}Q_j)\partial_{\dot{q}^i},
  \quad
  Q_i = \partial_{\dot{q}^i} C . 
\end{equation}
\end{cor}

\section{Dynamical symmetry groups}\label{sec:symmgroup}

Given an infinitesimal symmetry, 
it generates a one-parameter Lie group of symmetry transformations 
that maps the solution space of the equations of motion into itself. 
This mapping is given by the integral curve of the vertical vector field 
defining the infinitesimal symmetry,
which can be identified with the vector field \eqref{X.gen.solnsp} 
projected into the solution space $\solnsp$. 

The integral curve is obtained directly via the exponential mapping 
\begin{equation}\label{X.transformation.group}
  (q^i,\dot{q}^i)_\solnsp   \to  (q^i{}^*,\dot{q}^i{}^*)_\solnsp
  = \exp(\varepsilon\,\X_P^\solnsp) (q^i,\dot{q}^i)_\solnsp
\end{equation}
with parameter $\varepsilon$.
This defines a transformation group acting on the solution space $\solnsp$, 
where $\varepsilon=0$ yields the identity transformation
and where the underlying infinitesimal transformation is simply given by
the components of the vector field \eqref{X.gen.solnsp}: 
\begin{equation}
  (q^i{}^*,\dot{q}^i{}^*)_\solnsp
  = (q^i,\dot{q}^i)_\solnsp +\varepsilon \big(P^i(t,q,\dot{q})_\solnsp ,\Dt P^i(t,q,\dot{q})_\solnsp \big) + O(\varepsilon^2) . 
\end{equation}
An equivalent way to obtain the integral curve is by solving the ODE system
\begin{equation}\label{X.transformation.ODE}
  \frac{dq^i{}^*}{d\varepsilon} = \X_P^\solnsp (q^i){}^*,
  \quad
  \frac{d\dot{q}^i{}^*}{d\varepsilon} = \X_P^\solnsp (\dot{q}^i){}^*
\end{equation}
with initial conditions
$q^i{}^*|_{\varepsilon=0} = q^i$ and $\dot{q}^i{}^*|_{\varepsilon=0} = \dot{q}^i$. 

These symmetry transformations \eqref{X.transformation.group}
have a natural classification comprising two distinct types
\cite{BA-book}:
\emph{point symmetries} and \emph{dynamical symmetries}.
What distinguishes them is the dependence of the function $P^i$ on the variable $\dot{q}^i$. 

\begin{defn}\label{defn:point.vs.dyn}
A point symmetry generator $P^i_\text{pt.}(t,q,\dot{q})$ is
\emph{strictly linear in $\dot{q}^i$},
namely,
\begin{equation}\label{P.point}
  \dfrac{\partial P^i_\text{pt.}}{\partial \dot{q}^j} =-\tau(t,q) \delta_j{}^i
\end{equation}
where $\delta_j{}^i$ denotes the identity matrix. 
A dynamical symmetry generator $P^i_\text{dyn.}(t,q,\dot{q})$ is
\emph{either non-strictly linear or nonlinear in $\dot{q}^i$},
namely,
\begin{equation}\label{P.dyn}
  \dfrac{\partial P^i_\text{dyn.}}{\partial\dot{q}^j} = \sigma_j{}^i(t,q,\dot{q})
\end{equation}
with the matrix function $\sigma_j{}^i(t,q,\dot{q})$ satisfying
\begin{subequations}\label{P.dyn.cond}
\begin{equation}
  \parder{\sigma_j{}^i}{\dot{q}^k} =0,
  \quad 
  \sigma_j{}^i \neq -\tau  \delta_j{}^i , 
\end{equation}
or 
\begin{equation}
  \parder{\sigma_j{}^i}{\dot{q}^k} \neq 0 . 
\end{equation}
\end{subequations}
\end{defn}

The motivation for this classification comes from considering
the action of an infinitesimal symmetry on functions $(q^i(t),\dot{q}^i(t))$
which can be thought of as representing solution trajectories. 
Each solution geometrically describes a curve
in the coordinate space $(t,q^i,\dot{q}^i)$. 
Under a point transformation on $(t,q^i)$, given by 
\begin{equation}
  t \to t^\dagger = t + \varepsilon \tau(t,q) + O(\varepsilon^2),
  \quad
  q^i \to q^i{}^\dagger  = q^i + \varepsilon \eta^i(t,q) + O(\varepsilon^2)
\end{equation}
with
\begin{equation}
  \dot{q}^i \to \dot{q}^i{}^\dagger = \dot{q}^i + \varepsilon \eta_{(1)}^i(t,q,\dot{q}) + O(\varepsilon^2), 
  \quad
 \eta_{(1)}^i = \dot{\eta}^i -\dot{\tau}\dot{q}^i
\end{equation}
being the prolonged transformation on $(t,q^i,\dot{q}^i)$, 
a curve $(q^i(t),\dot{q}^i(t))$ is mapped to another curve
$(q^i{}^\dagger(t),\dot{q}^i{}^\dagger(t))$.
As illustrated in Fig.~1, 
this mapping of curves has a well-known equivalent formulation in which $t$ is kept fixed
and only $(q^i,\dot{q}^i)$ is changed \cite{BA-book,Olv-book}: 
\begin{equation}\label{point.transformation}
\begin{aligned}
  & q^i{}^*(t) = q^i(t) + \varepsilon \big(\eta^i(t,q(t)) - \tau(t,q(t))\,\dot{q}^i(t)\big) + O(\varepsilon^2),
  \\
  & \dot{q}^i{}^*(t) = \Dt q^i{}^*(t) =
  \dot{q}^i(t) + \varepsilon \big( \dot{\eta}^i(t,q(t)) - \dot{\tau}(t,q(t))\,\dot{q}^i(t) -\tau(t,q) f^i(t,q(t),\dot{q}(t)) \big) + O(\varepsilon^2) . 
\end{aligned}  
\end{equation}
Its corresponding generator is given by the vertical vector field 
\begin{equation}\label{point.X.gen}
  \X_\text{pt.}^\solnsp =   P_\text{pt.}^i \partial_{q^i} + \Dt P_\text{pt.}^i \partial_{\dot{q^i}}
\end{equation}
in the solution space, where
\begin{equation}
  P_\text{pt.}^i(t,q) = \eta^i(t,q) - \tau(t,q)\,\dot{q}^i
\end{equation}
has precisely the form satisfying the point symmetry definition \eqref{P.point}.

\begin{figure}[ht!]
\begin{tikzpicture}[scale=1.2]
\draw[->] (-1,0) -- (7,0) node[right] {$t$};
\draw[->] (0,-0.5) -- (0,3.5) node[left] {$q$};
\draw[thick] 
  (-1.2,1) .. controls (1.5,1.4) and (3,1.3) .. (5,2)
  node[right] {$q(t)$};
\draw[thick]
  (-1.2,3) .. controls (2,2.9) and (3.5,2.7) .. (5.8,3.3)
  node[right] {$q^{\dagger}(t)$};
\fill (2,1.35) circle (2pt);
\node[below right] at (2,1.35) {$(q,t)$};
\draw[dashed] (2,-0.5) -- (2,2.82);
\fill (2,2.87) circle (2pt);
\fill (3,2.89) circle (2pt);
\node[above right] at (3,2.89) {$(q^{\dagger},t^{\dagger})$};
\draw[->, thick] (2.1,1.45) .. controls (2.7,2.0) and (3,2.5) .. (3,2.8);
\draw[->, thick] (1.9,1.45) .. controls (1.7,1.9) and (1.7,2.5) .. (1.98,2.75);
\end{tikzpicture}
\caption{Point transformation of curves}
\end{figure}
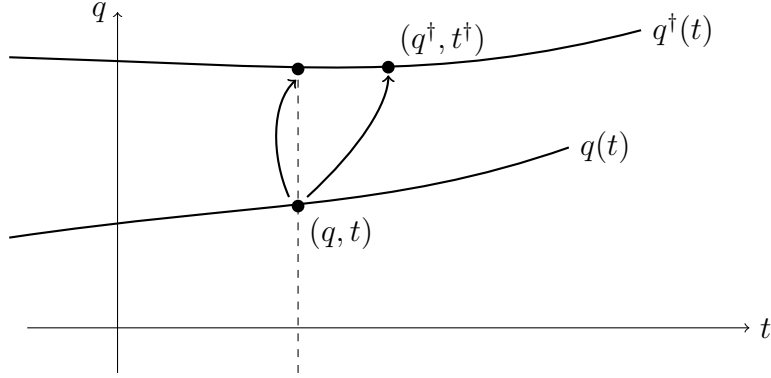

Therefore,
point symmetries arise from the restriction of
prolonged point transformation groups in the coordinate space $(t,q^i,\dot{q}^i)$
to the solution space $\solnsp$. 
By comparison,
dynamical symmetries cannot be expressed as the restriction to $\solnsp$ of
any transformation on the coordinate space $(t,q^i,\dot{q}^i)$.
Stated more generally, 
a dynamical symmetry cannot be lifted from $\solnsp$ to a transformation 
on spaces coordinatized by $t$, $q^i$, and time derivatives of $q^i$ up to a finite order
(cf. \cite{BCA-book}).

\subsection{Gauge freedom in symmetry generators} 

Lastly, 
an important observation is that any infinitesimal symmetry \eqref{X.gen.solnsp}
in the solution space $\solnsp$
can be extended to a vector field in the coordinate space $(t,q^i,\dot{q}^i)$. 
This extension is non-unique as it involves the following gauge freedom. 

Associate to the vertical vector field \eqref{X.gen.solnsp} 
a non-vertical vector field 
\begin{equation}\label{dyn.symm.Y}
  \Y_P = \X_P^\solnsp + \tau \Dt 
  = \tau\partial_{t} +\big(P^i + \tau\dot{q}^i\big)\partial_{q^i} + \big(\Dt P^i +\tau f^i\big)\partial_{\dot{q}^i}
\end{equation}
where $\tau$ is a completely arbitrary function of $t$, $q^i$, $\dot{q}^i$.
The integral curve of this vector field \eqref{dyn.symm.Y} in the solution space
yields the transformation group
\begin{equation}\label{dyn.symm.transformation} 
  (t,q^i,\dot{q}^i)_\solnsp \to  (t^\dagger,q^i{}^\dagger,\dot{q}^i{}^\dagger)_\solnsp
  = \exp(\varepsilon\,\Y) (t,q^i,\dot{q}^i)_\solnsp
\end{equation}
with parameter $\varepsilon$.

\begin{lem}\label{lem:Y.X}
The vector fields $\X_P^\solnsp$ and $\Y_P$ have the same action
on the solution space of the equations of motion.
In particular,
for any constant of motion $C(q,\dot{q})$, 
\begin{equation}\label{X.C.Y.C}
  \X_P^\solnsp(C)|_\solnsp = \Y_P(C)|_\solnsp . 
\end{equation}
\end{lem}

\begin{proof}
Consider just the vector field $\Dt$.
Its integral curve 
$(t,q^i,\dot{q}^i)_\solnsp \to  (t^\dagger,q^i{}^\dagger,\dot{q}^i{}^\dagger)_\solnsp
= \exp(\varepsilon\,\Dt) (t,q^i,\dot{q}^i)_\solnsp$
generates the transformations 
\begin{equation}
  t \to t^\dagger = t + \varepsilon, 
  \quad
  q^i \to q^i{}^\dagger  = q^i + \varepsilon \dot{q}^i + O(\varepsilon^2),
  \quad
  \dot{q}^i \to \dot{q}^i{}^\dagger = \dot{q}^i + \varepsilon f^i(t,q,\dot{q}) + O(\varepsilon^2) . 
\end{equation}
When acting on solutions,
the equivalent transformations with $t$ held fixed
are given by
$q^i{}^\dagger(t) = q^i(t) +O(\varepsilon^2)$
and $\dot{q}^i{}^\dagger(t) = \Dt q^i{}^\dagger(t) = \dot{q}^i(t)+O(\varepsilon^2)$
where the $O(\varepsilon)$ terms vanish, 
similarly to the computation giving the action of the point transformations \eqref{point.transformation}.
This implies that all higher order terms vanish, since they are determined uniquely
by the $O(\varepsilon)$ terms.  
Thus, $\Y = \Dt$ viewed as a symmetry vector field acts trivially on solutions,
and likewise for $\Y= \tau\Dt$.
Thereby, $\Y_P$ and $\X^\solnsp_P$ have the same action solutions
and in particular on any constant of motion. 
\end{proof}

Consequently, $\tau$ constitutes a \emph{gauge freedom}
in extending an infinitesimal symmetry $\X_P^\solnsp$ in $\solnsp$
to an equivalent symmetry vector field in the coordinate space $(t,q^i,\dot{q}^i)$. 
Existence of this gauge freedom has not been widely recognized or utilized in the literature on dynamical systems. 
It will turn out to enable finding an explicit form of
the symmetry transformation group for an infinitesimal dynamical symmetry,
alternatively to attempting to evaluate the exponential mapping expression \eqref{X.transformation.group} directly. 

\begin{cor}\label{cor:repr.X.transformation.group}
For an infinitesimal symmetry $\X_P^\solnsp$,
the transformation group \eqref{dyn.symm.transformation} for any choice of $\tau(t,q,\dot{q})$
provides a representation of the symmetry transformation group \eqref{X.transformation.group} generated by $\X_P^\solnsp$. 
\end{cor}

\section{Poisson brackets and action of symmetries on conserved quantities}\label{sec:PB.symmaction}

The set of all locally conserved quantities $C(t,q,\dot{q})$ of
a dynamical system \eqref{EL.eqn.gen}
is well known to possess several interrelated structures.
Firstly,
this set is a linear space,
since any linear combination of locally conserved quantities
is again a locally conserved quantity.
An arbitrary (smooth) function of any locally conserved quantity is also 
locally conserved.
Secondly,
as a linear space it carries an action of all symmetries of the equations of motion,
since any symmetry maps the solution space into itself,
whether the symmetry is variational or not. 
This linear space further comprises a Lie algebra,
with the commutator of two locally conserved quantities
being given by their Poisson bracket. 
Finally,
through the modern form of Noether's theorem (cf Theorem~\ref{thm:noether}), 
there is an homomorphism between this Poisson bracket Lie algebra 
and the variational symmetry Lie algebra 
consisting of the set of all infinitesimal symmetries,
where the Lie bracket is given by the commutator.

\subsection{Poisson bracket}

The Poisson bracket is conventionally defined in terms of
the phase space variables $(q^i,p_i)$
where 
\begin{equation}\label{Ham.p}
  p_i = \parder{L}{\dot{q}^i}
\end{equation}
defines the canonical momenta variables.
Changing variables from $(q^i,\dot{q}^i)$ to $(q^i,p_i)$
involves inverting the relation \eqref{Ham.p} to obtain
$\dot{q}^i$ in terms of $(q^i,p_j)$,
which is ensured by the non-degeneracy assumption $\det(g)\neq0$
on the Lagrangian $\Lagr(t,q,\dot{q})$ via the Hessian matrix \eqref{g.matrix}.
Under this transformation,
the Lagrangian equations of motion \eqref{eom.gen}
are equivalent to the Hamiltonian equations
\begin{equation}\label{Ham.eom.gen}
  \dot{q}^i = \parder{\Ham}{p_i},
  \quad
  \dot{p}_i = -\parder{\Ham}{q^i}
\end{equation}
where the Hamiltonian $\Ham$ is given by the Legendre transformation 
\begin{equation}\label{Ham.gen}
  \Ham(t,p,q) = \big( p_i \dot{q}^i  - \Lagr(t,q,\dot{q}) \big)\big|_{\dot{q}=\dot{q}(q,p)} . 
\end{equation}

It follows that the space of solutions $(q^i(t),p_i(t))$ of Hamilton's equations \eqref{Ham.eom.gen}
can be identified with the solution space $\solnsp$ of the Lagrangian equations of motion \eqref{eom.gen}.
This implies that a one-to-one correspondence holds between 
conserved quantities in the Lagrangian framework and the Hamiltonian framework. 

\begin{defn}\label{defn:conserved.Ham}
a function $\tilde C(t,q,p)$ on phase space is a \emph{locally conserved integral} (or invariant)
if
\begin{equation}
  \dfrac{d}{dt} \tilde C(t,q(t),p(t))\big|_{\solnsp} =0 
\end{equation}
holds piecewise in $t$ for every solution $(q^i(t),p_i(t))$ of Hamilton's equations \eqref{Ham.eom.gen}.
If for all solutions the conserved integral is continuous for all times $t$,
then it is \emph{globally conserved}. 
\end{defn}

This is the Hamiltonian counterpart of Definition~\ref{defn:conserved}.
The Legendre relation 
\begin{equation}\label{Legendre.conserved}
  \tilde C(t,q,p) = \tilde C(t,q,\partial_{\dot{q}}\Lagr) = C(t,q,\dot{q})
\end{equation}
gives the explicit correspondence between Lagrangian conserved integrals $C(t,q,\dot{q})$
and Hamiltonian conserved integrals $\tilde C(t,q,p)$. 

For any two functions $F_1(t,q,p)$ and $F_2(t,q,p)$ on phase space, 
their Poisson bracket is defined by 
\begin{equation}\label{Ham.PB}
  \{F_1,F_2\} = \parder{F_1}{q^i} \parder{F_2}{p_i} - \parder{F_2}{q^i} \parder{F_1}{p_i}
\end{equation}
which is bi-linear, antisymmetric, and obeys the Leibniz rule and the Jacobi identity.
This bracket is also non-degenerate, namely if 
$\{ F,G \}=0$ holds for all functions $G(t,q,p)$, then $F(t,q,p)$ is identically zero. 

As is well-known,
Hamilton's equations \eqref{Ham.eom.gen} have the Poisson bracket formulation
\begin{equation}\label{PB.eom.gen}
  \dot{q}^i = \{ \dot{q}^i, \Ham \}, 
  \quad
  \dot{p}_i = \{ p_i, \Ham \} . 
\end{equation}
Likewise,
the time evolution of any function $F(t,q,p)$ on phase space is given by 
\begin{equation}
  \dot{F} = \parder{F}{t} + \{ F, \Ham \} , 
\end{equation}
which follows directly from the change rule combined with Hamilton's equations.
Local conservation of an invariant $\tilde C(t,q,p)$ is consequently expressed as 
\begin{equation}\label{local.conserved.Ham}
  \parder{\tilde C}{t} + \{ \tilde C, \Ham \} =0 . 
\end{equation}
By means of the Legendre relation \eqref{Legendre.conserved},
the same condition holds for local conservation of $C(t,q,\dot{q})$. 
A further useful step is to formulation this condition
purely in terms of Lagrangian variables. 

\begin{lem}
Under a change of phase space variables to Lagrangian variables,
the Poisson bracket of any two functions $F_1$ and $F_2$ of $t$, $q$, $\dot{q}$
is given by
\begin{equation}\label{Lagr.PB}
  \{F_1,F_2\}  = \nabla^\t F_1\, \mathbf{J}\, \nabla F_2
  =   g^{-1}{}^{ij}  \Big( \parder{F_1}{q^i}\parder{F_2}{\dot{q}^j} -\parder{F_2}{q^i}\parder{F_1}{\dot{q}^j} \Big)
  +  c^{ij}  \parder{F_1}{\dot{q}^i}\parder{F_2}{\dot{q}^j} 
\end{equation}
where 
\begin{equation}\label{h.matrix}
  c^{ij} = - c^{ji} = g^{-1}{}^{ik}g^{-1}{}^{jl} (h_{ij} - h_{ji}),
  \quad
  h_{ij} = \parders{\Lagr}{q^i}{\dot{q}^j} . 
\end{equation}
Here $\nabla = \begin{pmatrix} \partial_{q^i} \\ \partial_{\dot{q}^i} \end{pmatrix}$
denotes the gradient in Lagrangian coordinates;
$\t$ denotes the transpose;
and 
\begin{equation}\label{PB.symplectic}
\mathbf{J} = 
\begin{pmatrix}
  0 & g^{-1}{}^{ij} \\ - g^{-1}{}^{ij} & c^{ij}
\end{pmatrix}
\end{equation}
denotes the symplectic matrix which is skew and obeys the Jacobi identity. 
\end{lem}

\begin{proof}
The Jacobian of the change of variables is the matrix 
\begin{equation}
  \dfrac{\partial(q^i,p_j)}{\partial(q^k,\dot{q}^l}
  = \begin{pmatrix} \delta_k{}^i & h_{kj} \\ 0 & g_{jl} \end{pmatrix} , 
\end{equation}
which has the inverse 
\begin{equation}
  \dfrac{\partial(q^k,\dot{q}^l)}{\partial(q^i,p_j}
  = \begin{pmatrix} \delta_i{}^k & -h_{im} g^{-1}{}^{lm} \\ 0 & g^{-1}{}^{lj} \end{pmatrix} . 
\end{equation}
Use of the chain rule given by the latter matrix applied to the phase space bracket \eqref{Ham.PB} directly yields the bracket \eqref{Lagr.PB}--\eqref{h.matrix}.

The properties of the symplectic matrix hold due to the corresponding properties
satisfied by the Poisson bracket \eqref{Ham.PB}. 
\end{proof}
        
This result more generally enables
carrying over the Hamiltonian Poisson bracket to the Lagrangian framework.
From a Lagrangian perspective,
this bracket \eqref{Lagr.PB} is an intrinsic coordinate-invariant structure. 
In particular, 
it is preserved under an arbitrary change of variables
$t\to \tilde t(t,q)$, $q^i\to \tilde q^i(t,q)$. 

As an immediate application of the Lagrangian Poisson bracket, 
the condition \eqref{local.conserved} for local conservation of a first integral
$C(t,q,\dot{q})$
has the direct reformulation
\begin{equation}\label{local.conserved.PB}
 \parder{C}{t} + \{ C, \tilde\Ham \} =0 
\end{equation}
where 
\begin{equation}
  \tilde\Ham(t,q,\dot{q}) = \dot{q}^i \parder{\Lagr(t,q,\dot{q})}{\dot{q}^i} - \Lagr(t,q,\dot{q}) 
\end{equation}
is the Hamiltonian \eqref{Ham.gen} expressed in Lagrangian variables.
Note that this condition \eqref{local.conserved.PB}
is equivalent to the conservation equation \eqref{local.conserved.Ham}
but has the advantage that the Legendre relation \eqref{Legendre.conserved} is not needed.

\subsection{Action of variational symmetries on conserved integrals}

Now consider the action of an infinitesimal variational symmetry \eqref{inv.cond.alt2.gen}
on a function of the Lagrangian variables. 
This action has a very natural formulation
in terms of the Lagrangian Poisson bracket \eqref{Lagr.PB}.

\begin{thm}\label{thm:X.PB.C.F} 
Let $C(t,q,\dot{q})$ be a locally conserved integral.
The corresponding infinitesimal variational symmetry \eqref{varsym.solnsp.gen}
projected into the solution space $\solnsp$
has the action 
\begin{equation}\label{X.PB.C.F}
  \X_{(C)}^\solnsp\rfloor dF = \{ F, C \}
\end{equation}
holding for any function $F(t,q,\dot{q})$. 
\end{thm}

\begin{proof}
Write the infinitesimal symmetry \eqref{varsym.solnsp.gen}
as $X_{(C)}^\solnsp = P^i \partial_{q^i} + \Dt P^i\partial_{\dot{q}^i}$.
Its components are given by 
\begin{equation}
  P^i = g^{-1}{}^{ij} \parder{C}{\dot{q}^j}
\end{equation}
and 
\begin{equation}
  \Dt{P}^i =
  \Dt g^{-1}{}^{ij} \parder{C}{\dot{q}^j}
  - g^{-1}{}^{ij} \Big( \parder{C}{q^j} + \parder{f^k}{\dot{q}^j}\parder{C}{\dot{q}^k}  \Big) 
\end{equation}
using the time-derivative relation \eqref{DtC} 
combined with the identity 
\begin{equation}
  \Big[ \frac{d}{dt}, \parder{}{\dot{q}^j} \Big] = - \parder{}{q^j} . 
\end{equation}
Hence, the left-hand side of equation \eqref{X.PB.C.F} is the expression
\begin{equation}
  \begin{aligned}
  \X_{(C)}^\solnsp\rfloor dF
  & = P^i\parder{F}{q^i} + \Dt P^i \parder{F}{\dot{q}^i} \\
  &  = g^{-1}{}^{ij}\Big(
    \parder{F}{q^i} \parder{C}{\dot{q}^j}
     - \parder{C}{q^i}  \parder{F}{\dot{q}^j}
     - \parder{f^k}{\dot{q}^j}\parder{C}{\dot{q}^k} \parder{F}{\dot{q}^i} \Big) 
    + \Dt g^{-1}{}^{ij} \parder{F}{\dot{q}^i} \parder{C}{\dot{q}^j} , 
  \end{aligned}
\end{equation}
which has used the symmetry of $g^{-1}{}^{ij}$,
while the right-hand side has the explicit form 
\begin{equation}
  \{ F,C \} = -\{ C,F \}
  =   g^{-1}{}^{ij}\Big( \parder{F}{q^i}\parder{C}{\dot{q}^j} -\parder{C}{q^i}\parder{F}{\dot{q}^j} \Big)
  +  c^{ij}  \parder{F}{\dot{q}^i}\parder{C}{\dot{q}^j} . 
  \end{equation}
Combining these expressions, and again using the symmetry of $g^{-1}{}^{ij}$,
yields 
\begin{equation}
  \{ C, F \}  + \X_{(C)}^\solnsp\rfloor dF =
  g^{-1}{}^{ki} g^{-1}{}^{lj} S_{kl}   \parder{C}{\dot{q}^i} \parder{F}{\dot{q}^j}
\end{equation}
where
\begin{equation}\label{S.expr}
  S_{kl}   =  g_{ki}g_{lj}\Dt g^{-1}{}^{ij} - g_{ik}\parder{f^i}{\dot{q}^l} +  h_{kl} - h_{lk}
\end{equation}
from equation \eqref{h.matrix}.

The first term in expression \eqref{S.expr} can be expanded to get 
\begin{equation}
g_{ki}g_{lj}\Dt g^{-1}{}^{ij}
= -\Dt g_{kl}
= -\parder{g_{kl}}{t}  -\dot{q}^j \parder{g_{kl}}{q^j} - f^j \parder{g_{kl}}{\dot{q}^j} .
\end{equation}
Substituting expression \eqref{f.eom} for $f^i$
into the second term in expression \eqref{S.expr}
gives
\begin{equation}
  g_{ik}\parder{f^i}{\dot{q}^l}
  = \parder{}{\dot{q}^l}\Big( \parder{L}{q^k} - \parders{\Lagr}{t}{\dot{q}^k}   -\dot{q}^j h_{jk} \Big)
  -f^i \parder{g_{ik}}{\dot{q}^l}
  = h_{kl} -h_{lk}   -\dot{q}^j \parder{h_{jk}}{\dot{q}^l}  -f^i \parder{g_{ik}}{\dot{q}^l}  - \parder{g_{lk}}{t} . 
\end{equation}
When combined with the last two terms in expression \eqref{S.expr},
this yields
\begin{equation}
  S_{kl}   =
  \dot{q}^j \Big( \parder{h_{jk}}{\dot{q}^l} - \parder{g_{kl}}{q^j} \Big) 
  + f^i\Big( \parder{g_{ik}}{\dot{q}^l} - \parder{g_{kl}}{\dot{q}^i} \Big)
   = 0 
\end{equation}
which vanishes due to commutativity of partial derivatives,
since 
\begin{equation}\label{commutativity.rels}
  \parder{h_{jk}}{\dot{q}^l} 
  = \dfrac{\partial^3 \Lagr}{\partial q^j\partial \dot{q}^k\partial \dot{q}^l}
  = \parder{g_{kl}}{q^j}, 
  \quad
  \parder{g_{ik}}{\dot{q}^l} 
  = \dfrac{\partial^3 \Lagr}{\partial \dot{q}^i\partial \dot{q}^k\partial \dot{q}^l}
  = \parder{g_{kl}}{\dot{q}^i} . 
\end{equation}
\end{proof}

This theorem essentially states a Lagrangian counterpart of
how a time-dependent Hamiltonian vector field acts on functions on phase space.
Often in Hamilton mechanics, a Hamilton vector field is defined to be time-independent,
although the same definition carries over to the time-dependent case without change,
as evidenced by the correspondence \eqref{X.PB.C.F}. 

A useful corollary is obtained when the function is a conserved integral.

\begin{cor}\label{cor:X.PB.C1.C2}
Let $C_1(t,q,\dot{q})$ and $C_2(t,q,\dot{q})$ 
be locally conserved integrals. 
Their corresponding infinitesimal variational symmetries \eqref{varsym.solnsp.gen}
projected into the solution space $\solnsp$
satisfy the relation 
\begin{equation}\label{X.PB.C1.C2}
  \X^\solnsp_{(C_1)}\rfloor dC_2 
  = -\X^\solnsp_{(C_2)}\rfloor dC_1
  = \{ C_2, C_1 \} . 
\end{equation}
\end{cor}

The preceding results can be used to give a Poisson bracket formula
for the symmetry transformation group generated by
an infinitesimal variational symmetry \eqref{varsym.solnsp.gen}.

\begin{prop}\label{prop:varsymm.transformation.group}
Let $\{\ \cdot\ ,F\}^n$ denote the $n$-times nested Poisson bracket
$\underbrace{\{\ldots,\{\ \cdot\ ,F\},\ldots,F\}}_{\text{$n$ times}}$
for any function $F(t,q,\dot{q})$. 
The symmetry transformation group generated by 
an infinitesimal variational symmetry \eqref{varsym.solnsp.gen}
arising from a locally conserved integral $C(t,q,\dot{q})$
can be expressed in terms of the Poisson bracket as 
\begin{equation}
  \exp\big(\varepsilon \X_{(C)}^\solnsp\big)(q^i,\dot{q}^i) 
  = \exp\big(\varepsilon\{\ \cdot\ ,C\}\big)(q^i,\dot{q}^i) 
  = (q^i,\dot{q}^i) + \sum_{n=1}^{\infty} \frac{\varepsilon^n}{n!} (\{q^i,C\}^n,\{\dot{q}^i,C\}^n)
\end{equation}
with group parameter $\varepsilon$. 
\end{prop}

Corollary~\ref{cor:X.PB.C1.C2} has an important computational use. 
If the Poisson brackets among a set of conserved quantities are known,
then relation \eqref{X.PB.C1.C2} directly yields 
the action of the corresponding infinitesimal symmetries on the conserved quantities.
Alternatively, if the symmetry actions are known,
then they can be used to obtain the Poisson brackets
from relation \eqref{X.PB.C1.C2}.
This is particularly useful when the symmetries have a geometrical meaning
which allow their action to be determined entirely by geometric considerations. 
This will be illustrated in the examples in forthcoming work \cite{Anc2026}.

\subsection{Lie algebras of variational symmetries and conserved integrals}

The commutator of infinitesimal variational symmetries can be formulated
in terms of Poisson brackets of the locally conserved integrals
arising from Noether's theorem (cf Theorem~\ref{thm:noether}).

\begin{thm}\label{thm:C1.C2.varsymm.PB} 
For any locally conserved integrals $C_1(t,q,\dot{q})$, $C_2(t,q,\dot{q})$,
the commutator of their corresponding infinitesimal variational symmetries \eqref{varsym.solnsp.gen}
projected into the solution space $\solnsp$ is given by 
 \begin{equation}\label{C1.C2.varsymm.PB}
    [\X_{(C_2)}^\solnsp,\X_{(C_1)}^\solnsp] = \X_{(\{ C_1, C_2 \})}^\solnsp
 \end{equation}
using their Poisson bracket. 
\end{thm}

The counterpart of this result in Hamiltonian mechanics (see e.g. \Ref{GolPooSaf,Arn-book}) 
is well known in the case when the conserved integrals are constants of motion.
A special case when the conserved integrals arise from point symmetries of the Hamiltonian is also widely known.
These results actually hold in much more generality, as stated in Theorem~\ref{thm:C1.C2.varsymm.PB}.
A direct proof will be given in the Appendix. 

A main corollary of this theorem is that
there is a homomorphism relating
the Lie algebra of infinitesimal variational symmetries
and Lie algebra of locally conserved integrals. 

\begin{cor}\label{cor:liealgebras}
The Lie algebra of all locally conserved integrals under the Poisson bracket
is homomorphic to
the Lie algebra of all infinitesimal variational symmetries \eqref{varsym.solnsp.gen} 
under commutation, 
where the homomorphism uses the correspondence \eqref{PfromC.gen}
given by Noether's theorem (cf Theorem~\ref{thm:noether}). 
\end{cor}  

The kernel of the homomorphism consists of constants, which are trivially conserved.
They are mapped into the zero symmetry by the Noether correspondence. 

To understand the homomorphism fully,
it is necessary to consider the notion of independence for locally conserved integrals
as well as infinitesimal variational symmetries. 
Recall that two locally conserved integrals $C_1$ and $C_2$
are \emph{functionally independent}
if $k(C_1,C_2)=0$ holds only for the trivial function $k\equiv0$.
The Noether correspondence \eqref{PfromC.gen}
leads to the following counterpart notion of independence for variational symmetries.

\begin{defn}\label{defn:X.dependent}
Two variational symmetries $\X^\solnsp_{(C_1)}$ and $\X^\solnsp_{(C_2)}$
are \emph{independent over the solution space}
if and only if 
\begin{equation}\label{X.dependent}
  k_{C_1} \X^\solnsp_{(C_1)} + k_{C_2} \X^\solnsp_{(C_2)} =0
\end{equation}
holds only when the function $k(C_1,C_2)\equiv0$ is trivial. 
\end{defn}

Note that this is a stronger condition than linear independence,
since two variational symmetries could be linearly independent, 
$c_1\X^\solnsp_{(C_1)} + c_2\X^\solnsp_{(C_2)} \neq 0$ for all constants $c_1$ and $c_2$,
but still be dependent over the solution space \eqref{X.dependent}.
In particular, the variational symmetries
$\X^\solnsp_{(C_1)}$ and $\X^\solnsp_{(k(C_1))} = k'(C_1)\X^\solnsp_{(C_1)}$
are linearly independent when the function $k$ is nonlinear, $k''\neq0$, 
yet they are dependent over the solution space.

The following result is a straightforward consequence of the preceding discussion. 

\begin{prop}\label{prop:count}
For a dynamical system \eqref{EL.eqn.gen} with $N$ degrees of freedom:
(i) A set of locally conserved integrals is functionally independent
if and only if the corresponding variational symmetries
are independent over the solution space.
(ii) The number of functionally-independent locally conserved integrals
is $2N$,
which is equal to the number of variational symmetries that are independent over the solution space.
(iii) The total number of linearly independent variational symmetries is infinite.
(iv) A set of locally conserved integrals is homomorphic to a finite-dimensional Lie algebra of variational symmetries
if and only if their Poisson brackets close linearly.
(v) If a variational symmetry $\X^\solnsp_{(C)}$ is a point symmetry,
then the variational symmetries $\X^\solnsp_{(k(C))} =k'(C)\X^\solnsp_{(C)}$
for any nonlinear function $k(C)$ are dynamical symmetries.
\end{prop}

There is an important computational use for
Theorem~\ref{thm:C1.C2.varsymm.PB} together with Corollary~\ref{cor:liealgebras}. 
If the commutators of a set of infinitesimal symmetries are known,
then relation \eqref{C1.C2.varsymm.PB} directly yields the corresponding Poisson brackets. 
Alternatively, if the Poisson brackets are known, 
then relation \eqref{C1.C2.varsymm.PB} provides the commutators of the corresponding infinitesimal symmetries. 
This will be illustrated in the examples in forthcoming work \cite{Anc2026}.

\section{Liouville integrable systems}\label{sec:liouville.system}

The main results stated in
Theorems~\ref{thm:noether} to~\ref{thm:C1.C2.varsymm.PB},
as well as in Corollaries~\ref{cor:repr.X.transformation.group} to~\ref{cor:liealgebras},
will now be applied to dynamical systems that are locally Liouville integrable.

Recall \cite{Arn-book} that an autonomous Hamilton system
exhibits global Liouville integrability 
when the number of commuting constants of motion 
is the same as the number of degrees of freedom. 
This allows the introduction of action-angle variables 
for explicitly integrating the equations of motion
provided that all of the constants of motion are globally conserved
on every solution trajectory. 
An extension of this result holds for non-autonomous Hamiltonian systems \cite{Bou.Bou}
that possess commuting globally conserved integrals,
which may be integrals of motion or constants of motion. 

Existence of action-angle variables does not itself rely on any global conditions,
and hence a local version of Liouville integrability can be formulated
for dynamical systems \eqref{EL.eqn.gen} with $N$ degrees of freedom 
by using $N$ locally conserved integrals that are commuting. 
One primary objective here is to see how the resulting action-angle variables
lead to expressions for additional locally conserved integrals,
which will allow for explicit integration of the equations of motion locally in time. 
In particular, there will be a total of $2N$ locally conserved integrals.
Through Noether's theorem, this yields $2N$ variational symmetries,
each of which will generate a one-dimensional Lie group of symmetry transformations.
These transformation groups act in the space of Lagrangian variables, 
or equivalently in phase space.
They have an equivalent formulation, involving gauge freedom,
in an extended space of variables when time is adjoined as a variable. 
This gauge freedom is shown to be useful for obtaining
an explicit form for the transformations. 

It is important to emphasize that the extra $N$ conserved integrals
arising from the action-angle variables 
may be non-commuting, or may be only piecewise continuous on solution trajectories. 
(See \Ref{Ley} for examples.)

\subsection{Conserved integrals from action-angle variables}

Suppose a dynamical system \eqref{EL.eqn.gen}
possesses $N$ functionally independent locally conserved integrals 
$C_i(q,\dot{q})$, $i=1,\dots, N$, 
whose Poisson brackets vanish, $\{C_i,C_j\}=0$.
They can be used to define a generating function $G(t,q,C)$
by requiring that
\begin{equation}\label{G.eqn}
  \dfrac{\partial G}{\partial q^i}=p_i
\end{equation} 
yields the Hamiltonian momenta \eqref{Ham.p}
given in terms of the Lagrangian  $\Lagr(t,q,\dot{q})$,
where $\dot{q}^i$ is expressed as a function of $t,q^j,C_j$
by inverting the expressions $C_i=C_i(t,q,\dot{q})$
through use of the implicit function theorem. 
Integration of equation \eqref{G.eqn} then gives 
\begin{equation}\label{G.funct}
G = \int p_i(t,q,C)\,dq^i, 
\quad
p_i = \parder{\Lagr}{\dot{q}^i}
\end{equation}
where the constant of integration, which is an arbitrary function of $C_j$ and $t$,
has been put to zero for simplicity. 

The generating function \eqref{G.funct} gives rise to a canonical transformation
\begin{equation}\label{G.transformation}
  (q^i,p_i) \to (\Theta^i,C_i)
\end{equation}
which locally preserves the Hamiltonian form \eqref{Ham.eom.gen}
of the equations of motion, 
with the new canonical momentum variables being the commuting conserved integrals,
and with the new position variables defined as
\begin{equation}\label{angle.variables}
  \Theta^i = \dfrac{\partial G}{\partial C_i} . 
\end{equation}
The transformed equations of motion are given by 
\begin{equation}\label{angle.eom}
\dot{\Theta}^i =\parder{K}{C_i}, 
\quad
\dot{C}_i = -\parder{K}{\Theta^i} =0
\end{equation}
with
\begin{equation}
  K = H + \parder{G}{t}
\end{equation}
(see e.g. \cite{GolPooSaf})
where $H=H(t,\Theta,C)$ is the original Hamiltonian \eqref{Ham.gen}
expressed in terms of the new canonical variables $(\Theta^i,C_i)$ and $t$. 

This transformation \eqref{G.transformation}
amounts to introducing action-angle variables,
since local conservation $\dot{C}_i=0$ implies that
$K$ has no dependence on $\Theta^i$,
whereby $\dot{\Theta}^i$ is just a function of $C_i$ and $t$.
As a consequence, 
$\Theta^i$ can be obtained in an explicit form by integration of
$\dfrac{\partial K}{\partial C_i}$ with respect to $t$,
which together with the commuting conserved integrals gives
a total of $2N$ conserved integrals that represent
the explicit solution of the system locally in time. 
Namely, the system is locally Liouville integrable.
This generalizes the well-known result \cite{GolPooSaf,Arn-book}
that holds for integrable autonomous Hamiltonian systems.

Specifically, when $\dfrac{\partial H}{\partial t}=0$,
then $G$ has no explicit dependence on $t$,
and hence $K =H(C)$.
Since $H=E$ is the energy constant of motion of the system,
it can be taken as one of the $N$ starting conserved integrals, say $C_1=E$.
Consequently, $K=E$ yields 
\begin{equation}
\dot{\Theta}^1 = 1,
\quad
\dot{\Theta}^2 = \ldots = \dot{\Theta}^N =0 , 
\end{equation}
from the angle equations \eqref{angle.eom}, 
whereby $\Theta^2,\ldots, \Theta^N$ are $N-1$ locally conserved constants of motion,
while $t - \Theta^1$ is a locally conserved temporal integral of motion. 

In the case when $\dfrac{\partial H}{\partial t}\neq0$,
a similar statement holds by the following argument.
The angles \eqref{angle.variables} give integrals of motion 
$A^i = \Theta^i - \bigint \dfrac{\partial K}{\partial C_i}\,dt$,
$i=1,\ldots,N$,
which are locally conserved, $\dot{A}^i=0$,
due to equations \eqref{angle.eom}. 
Now consider the linear homogeneous partial differential equation
$0= \partial_t F(A^1,\ldots,A^N) = -\dfrac{\partial K}{\partial C_i} \partial_{A^i} F$.
It admits $N-1$ particular solutions $F_1(A),\ldots,F_{N-1}(A)$,
each of which has no explicit dependence on $t$.
Thus, these solutions represent $N-1$ locally conserved constants of motion.
Any one of $A_1,\ldots,A_N$ provides a locally conserved integral of motion. 
This constitutes local Liouville integrability for non-autonomous Hamiltonian systems. 

A purely Lagrangian formulation of local integrability will now be stated. 

\begin{thm}\label{thm:liouville}
For a dynamical system \eqref{EL.eqn.gen} with 
$N$ locally conserved integrals $C_i(q,\dot{q})$, $i=1,\dots, N$,
that are functionally independent and have vanishing Poisson brackets, 
let 
\begin{equation} 
\mathcal{S}(t,q,C) = \int \parder{\Lagr}{\dot{q}^j}(t,q,\dot{q}(t,q,C))\, dq^j
\end{equation}
where $\dot{q}^j(t,q,C)$ is given by inverting $C_i=C_i(t,q,\dot{q})$. 
The generalized angle variables 
\begin{equation} 
\Theta^i = \parder{\mathcal{S}}{C_i} = \int g_{jk} \partial_{C_i}\dot{q}^k(t,q,C)\, dq^j
\end{equation}
lead to $N$ additional locally conserved integrals,
where $g_{jk}$ is the Hessian matrix \eqref{g.matrix}.
In particular:\\
(i) if the Lagrangian is autonomous, $\partial_t \Lagr =0$,
then by choosing
\begin{equation}
  C_1 = H = \dot{q}^i \partial_{\dot{q}^i}\Lagr -\Lagr,
\end{equation}
the generalized angles $\Theta^2,\ldots, \Theta^N$ constitute $N-1$ constants of motion,
and $T= t - \Theta^1$ is a temporal integral of motion.\\
(ii) If the Lagrangian is non-autonomous, $\partial_t \Lagr \neq0$,
then 
\begin{equation}\label{temporal.integral}
  \Upsilon^i=  \int (\partial_t\Theta^i + \partial_{q^j}\Theta^i \dot{q}^j(t,q,C)))\, dt - \Theta^i
\end{equation}
are integrals of motion,
in terms of which $N-1$ constants of motion are given by
solutions of the linear homogeneous partial differential equation 
$\dfrac{\partial K}{\partial C_i} \partial_{\Upsilon^i} F(\Upsilon^1,\ldots,\Upsilon^N)=0$ 
where $K = \dot{\mathcal{S}}(t,q,C)  -\Lagr(t,q,\dot{q}(t,q,C))$. 
\end{thm}

The proof essentially consists of transcribing
the generating function \eqref{G.funct}
and the definition of the angle variables \eqref{angle.variables}
into Lagrangian form,
combined with use of the chain rule
$\partial_{C_i}= \partial_{C_i}\dot{q}^j(t,q,C)\partial_{\dot{q}^j}$.

\subsection{Variational symmetries}

Through the Noether correspondence \eqref{PfromC.gen},
the $N$ locally conserved integrals $C_i(t,q,\dot{q})$
correspond to $N$ variational symmetries
\begin{subequations}\label{liouville.C.X}
\begin{equation}
\X^\solnsp_{(C_i)} = P^j_{(C_i)} \partial_{q_j} + \Dt P^j_{(C_i)} \partial_{\dot{q}_j}
\end{equation}
where
\begin{equation} 
P^j_{(C_i)} = g^{-1}{}^{jk}\partial_{\dot{q}^k} C_i . 
\end{equation}
\end{subequations}
Their commutators are given by
\begin{equation}
  [\X^\solnsp_{(C_i)}, \X^\solnsp_{(C_j)}] =\X^\solnsp_{(\{C_j,C_i\})} =0
\end{equation}
from Theorem~\ref{thm:C1.C2.varsymm.PB},
since $\{C_j,C_i\} =0$. 
Hence, the set of these infinitesimal symmetries \eqref{liouville.C.X}
forms an $N$-dimensional abelian Lie algebra. 

Conversely, when a dynamical system possesses
$N$ commuting variational symmetries that are independent over the solution space, 
the Noether correspondence \eqref{PfromC.gen} yields $N$ locally conserved integrals.
These integrals will be functionally independent
by part (i) of Proposition~\ref{prop:count}.
Their Poisson brackets will be constants,
due to the homomorphism \eqref{C1.C2.varsymm.PB}
given in Theorem~\ref{thm:C1.C2.varsymm.PB}.
However, these brackets need not vanish,
which implies that existence of
commuting variational symmetries over the solution space
is insufficient for local Liouville integrability to hold.

From Theorem~\ref{thm:liouville},
the $N$ additional locally conserved integrals 
likewise correspond to $N$ variational symmetries.
Specifically,
in the case of an autonomous dynamical system,
the $N-1$ constants of motion $\Theta^i(q,\dot{q})$, $i=2,\ldots,N$,
yield 
\begin{subequations}\label{liouville.Theta.X}
\begin{equation}
  \X^\solnsp_{(\Theta^i)} = P^j_{(\Theta^i)} \partial_{q_j} + \Dt P^j_{(\Theta^i)} \partial_{\dot{q}_j},
  \quad
  i=2,\ldots,N
\end{equation}
where
\begin{equation}
P^j_{(\Theta^i)} = g^{-1}{}^{jk}\partial_{\dot{q}^k} \Theta^i
= g^{-1}{}^{jk}\partial_{\dot{q}^k}C_l  \parders{\mathcal{S}}{C_l}{C_i} . 
\end{equation}
\end{subequations}
The integral of motion $T$ similarly yields
\begin{subequations}\label{liouville.T.X}
\begin{equation}
\X^\solnsp_{(T)} = P^j_{(T)} \partial_{q_j} + \Dt P^j_{(T)} \partial_{\dot{q}_j}
\end{equation}
with 
\begin{equation}
P^j_{(T)} = -g^{-1}{}^{jk}\partial_{\dot{q}^k} \Theta^1
= -g^{-1}{}^{jk}\partial_{\dot{q}^k}C_l  \parders{\mathcal{S}}{C_l}{E}
\end{equation}
\end{subequations}
where $C_1=H=E$ is the energy. 
The commutators $[\X^\solnsp_{(\Theta^i)},\X^\solnsp_{(\Theta^j)}]$ 
and $[\X^\solnsp_{(\Theta^i)},\X^\solnsp_{(C^j)}]$ 
are given by variational symmetries that arise from the corresponding Poisson brackets. 
Note that these Poisson brackets may not close linearly,
for instance $\{\Theta^i,\Theta^j\}$ may be a nonlinear function of 
$\Theta^i$, $T$, $C_i$.
In such a situation,
the set of variational symmetries \eqref{liouville.Theta.X} and \eqref{liouville.T.X}
will not close under commutation and thereby does not form a Lie algebra.
However, the resulting variational symmetries arising from such commutators
will be dependent over the solution space,
as a consequence of part (iii) of Proposition~\ref{prop:count}.

A similar discussion applies to the conserved integrals in Theorem~\ref{thm:liouville}
in the case of non-autonomous dynamical systems. 

Whether or not a dynamical system is autonomous,
each of the $2N$ variational symmetries in Theorem~\ref{thm:liouville} 
generates a one-dimensional Lie group of symmetry transformations, 
which can be obtained by Corollary~\ref{cor:repr.X.transformation.group}.
The symmetry transformations that correspond to the $N$ conserved integrals $C_i$
will commute, giving a $N$-dimensional abelian Lie group,
while closure of the remaining symmetry transformations
depends on the Poisson bracket structure of the corresponding conserved integrals.

\begin{prop}\label{prop:closed.symmgroup}
For a dynamical system that is locally Liouville integrable,
the $2N$ variational symmetries in Theorem~\ref{thm:liouville}
will generate a Lie group of closed symmetry transformations
if and only if the Poisson brackets of the $2N$ conserved integrals close linearly.
\end{prop}

A final remark is that if a dynamical system is autonomous
then the constant of motion $C_1=H=E$ representing energy
turns out to corresponds to a time-translation symmetry.
This well-known statement is derived as follows
from the variational symmetry
\begin{equation} 
\X^\solnsp_{(E)} = P^j_{(E)} \partial_{q_j} + \Dt P^j_{(E)} \partial_{\dot{q}_j}
\end{equation}
where $P^j_{(E)} = g^{-1}{}^{jk}\partial_{\dot{q}^k} E$ is given in terms of $E$. 
Applying the chain rule to the derivative in this latter expression yields
\begin{equation} 
  \partial_{\dot{q}^k} E
  = \partial_{\dot{q}^k} p_l \partial_{p_l} E
  = g_{kl} \dot{q}^l
\end{equation}
by use of the expression \eqref{Ham.p} for the Hamiltonian momenta
combined with Hamilton's equations \eqref{Ham.eom.gen}. 
This gives
\begin{equation}
  P^j_{(E)} = \dot{q}^j , 
\end{equation}
and thus
\begin{equation} 
\X^\solnsp_{(E)} = \dot{q}^j  \partial_{q_j} + f^j\partial_{\dot{q}_j}
\end{equation}
using the equations of motion \eqref{eom.gen}.
Now consider the associated vector field
\begin{equation}
  \Y_{(E)} = \X^\solnsp_{(E)} + \tau \Dt 
    = \tau\partial_{t} +(\tau +1)\dot{q}^i \partial_{q^i} + (\tau +1)f^i\partial_{\dot{q}^i}
\end{equation}
and choose the gauge function $\tau=-1$.
This yields 
\begin{equation}\label{E.Y}
  \Y_{(E)} = - \partial_{t} , 
\end{equation}
which is an infinitesimal time-translation point symmetry.
It generates the Lie group of symmetry transformations
$t^\dagger = t - \varepsilon$,
with parameter $\varepsilon\in\Rnum$,
as seen by Corollary~\ref{cor:repr.X.transformation.group}.

\section{Concluding remarks}\label{sec:conclude}

In the present work, a hybrid framework has been developed that highlights and unifies
the most important aspects of the Noether correspondence between
symmetries and conserved integrals in Lagrangian and Hamiltonian mechanics.
The results enable, in particular,
finding the complete Noether symmetry group of a dynamical system 
in an explicit form. 
This will be illustrated in a sequel paper \cite{Anc2026}
that will discuss the dynamical symmetries
for several examples of dynamical systems which satisfy local Liouville integrability.

\section*{Acknowledgements}

S.C.A.\ is supported by an NSERC research grant.

\appendix
\section{Proof of Theorem~\ref{thm:C1.C2.varsymm.PB}}

Because infinitesimal variational symmetries form a Lie algebra under commutation, 
the left-hand side of equation \eqref{C1.C2.varsymm.PB} is a vertical vector field
$P^i_\text{l.h.s.}\partial_{q^i} + \Dt P^i_\text{l.h.s.}\partial_{\dot{q}^i}$ where
\begin{equation}\label{P.lhs}
  P^i_\text{l.h.s.} =
  \X_{(C_2)}^\solnsp\rfloor d P_1^i   - \X_{(C_1)}^\solnsp\rfloor d P_2^i 
\end{equation}
with $P_1^i = g^{-1}{}^{ij} \dfrac{\partial C_1}{\partial \dot{q}^j}$
and $P_2^i = g^{-1}{}^{ij} \dfrac{\partial C_2}{\partial \dot{q}^j}$. 
Using Theorem~\ref{thm:X.PB.C.F}, 
the first term in expression \eqref{P.lhs} is given by 
\begin{equation}
  \begin{aligned}
  \X_{(C_2)}^\solnsp\rfloor d P_1^i 
  = \{ P_1^i, C_2 \}
  & = g^{-1}{}^{kl}\Big( \parder{P_1^i}{q^k} \parder{C_2}{\dot{q}^l}
  -\parder{C_2}{q^k}\parder{P_1^i}{\dot{q}^l} \Big)
  +  c^{kl}  \parder{P_1^i}{\dot{q}^k}\parder{C_2}{\dot{q}^l}  \\
  & =  P_2^k  \parder{P_1^i}{q^k}
  - g^{-1}{}^{kl} \parder{C_2}{q^k}\parder{P_1^i}{\dot{q}^l}
  + c^{kj} g_{jl} P_2^l   \parder{P_1^i}{\dot{q}^k} 
  \end{aligned}
\end{equation}
after using equation \eqref{h.matrix}. 
The second term in expression \eqref{P.lhs} has a similar expansion,
and thus 
\begin{equation}\label{P.lhs.simp}
  P^i_\text{l.h.s.} = 
  P_2^k  \parder{P_1^i}{q^k} - P_1^k  \parder{P_2^i}{q^k} 
  + g^{-1}{}^{kl} \Big( \parder{C_1}{q^k}\parder{P_2^i}{\dot{q}^l}
  - \parder{C_2}{q^k}\parder{P_1^i}{\dot{q}^l} \Big)
  + c^{kj} g_{jl}\Big( P_2^l   \parder{P_1^i}{\dot{q}^k} - P_1^l   \parder{P_2^i}{\dot{q}^k} \Big) . 
\end{equation}

The right-hand side of equation \eqref{C1.C2.varsymm.PB} is a vertical vector field
$P^i_\text{r.h.s.}\partial_{q^i} + \Dt P^i_\text{r.h.s.}\partial_{\dot{q}^i}$
where, from the Poisson bracket \eqref{Lagr.PB},
\begin{equation}\label{P.rhs}
  \begin{aligned}
  P^i_\text{r.h.s.} 
  & = g^{-1}{}^{ij} \parder{}{\dot{q}^j} \Big( g^{-1}{}^{kl}
  \Big( \parder{C_1}{q^k}\parder{C_2}{\dot{q}^l} -\parder{C_2}{q^k}\parder{C_1}{\dot{q}^l} \Big)
  +  c^{kl} \parder{C_1}{\dot{q}^k}\parder{C_2}{\dot{q}^l} 
  \Big) \\
  & = g^{-1}{}^{ij} \parder{}{\dot{q}^j} \Big( P_2^k \parder{C_1}{q^k} - P_1^k \parder{C_2}{q^k}
  + (h_{kl}-h_{lk}) P_1^k P_2^l \Big)
  \end{aligned}
\end{equation}
with use of equation \eqref{h.matrix}. 
Expanding the first two terms in this expression produces
\begin{equation}
g^{-1}{}^{ij} \parder{}{\dot{q}^j} \Big( P_2^k \parder{C_1}{q^k} - P_1^k \parder{C_2}{q^k} \Big) 
  =   g^{-1}{}^{ij} \Big( \parder{P_2^k}{\dot{q}^j} \parder{C_1}{q^k}
  - \parder{P_2^k}{\dot{q}^j} \parder{C_1}{q^k} \Big)
  + g^{-1}{}^{ij} \Big( P_2^k \parders{C_1}{\dot{q}^j}{q^k}  - P_1^k \parders{C_2}{\dot{q}^j}{q^k} \Big) 
\end{equation}
where the last term in parentheses can be expressed through integration by parts as 
\begin{equation}
  \begin{aligned}
   g^{-1}{}^{ij} \Big( P_2^k \parders{C_1}{\dot{q}^j}{q^k}  - P_1^k \parders{C_2}{\dot{q}^j}{q^k} \Big)
   & = P_2^k \parder{P_1^i}{q^k}  - P_1^k \parder{P_2^i}{q^k}
   - \parder{g^{-1}{}^{ij}}{q^k}  \Big( P_2^k \parder{C_1}{\dot{q}^j}  - P_1^k \parder{C_2}{\dot{q}^j} \Big) \\
   & = P_2^k \parder{P_1^i}{q^k}  - P_1^k \parder{P_2^i}{q^k}
   + g^{-1}{}^{il} \parder{g_{lj}}{q^k}  \Big( P_2^k P_1^j  - P_1^k P_2^j \Big) . 
  \end{aligned}
\end{equation}
Next, the second two terms in expression \eqref{P.rhs}
can be expanding and rearranged,
giving 
\begin{equation}
  g^{-1}{}^{ij} \parder{}{\dot{q}^j} \Big(  (h_{kl}-h_{lk}) P_1^k P_2^l \Big) 
  = (h_{kl}-h_{lk}) g^{-1}{}^{ij} \parder{}{\dot{q}^j} \Big( P_1^k P_2^l \Big)
  + g^{-1}{}^{ij} \parder{h_{kl}}{\dot{q}^j} \big( P_1^k P_2^l - P_2^k P_1^l \big) . 
\end{equation}
This yields 
\begin{equation}
  \begin{aligned}
  P^i_\text{r.h.s.}
   & =  g^{-1}{}^{ij} \Big( \parder{P_2^k}{\dot{q}^j} \parder{C_1}{q^k}
  - \parder{P_2^k}{\dot{q}^j} \parder{C_1}{q^k} \Big) 
  + P_2^k \parder{P_1^i}{q^k}  - P_1^k \parder{P_2^i}{q^k}
  \\&\qquad
  +  (h_{kl}-h_{lk}) g^{-1}{}^{ij} \parder{}{\dot{q}^j} \Big( P_1^k P_2^l \Big)
  + g^{-1}{}^{ij} \Big( \parder{h_{kl}}{\dot{q}^j}   -\parder{g_{jl}}{q^k}  \Big)
  \big( P_1^k P_2^l - P_2^k P_1^l \big) . 
  \end{aligned}
\end{equation}
The last term vanishes due to the first commutativity relation in equation \eqref{commutativity.rels}, 
and thus
\begin{equation}\label{P.rhs.simp}
  P^i_\text{r.h.s.}
   =  g^{-1}{}^{ij} \Big( \parder{P_2^k}{\dot{q}^j} \parder{C_1}{q^k}
  - \parder{P_2^k}{\dot{q}^j} \parder{C_1}{q^k} \Big) 
  + P_2^k \parder{P_1^i}{q^k}  - P_1^k \parder{P_2^i}{q^k}
  +  (h_{kl}-h_{lk}) g^{-1}{}^{ij} \parder{}{\dot{q}^j} \Big( P_1^k P_2^l \Big) . 
\end{equation}

Now, subtract expressions \eqref{P.rhs.simp} and \eqref{P.lhs.simp}: 
\begin{equation}\label{P.lhs.rhs}
  \begin{aligned}
  P^i_\text{l.h.s.} -    P^i_\text{r.h.s.}
  & = \parder{C_1}{q^k}\Big( g^{-1}{}^{kl} \parder{P_2^i}{\dot{q}^l}
  - g^{-1}{}^{ij} \parder{P_2^k}{\dot{q}^j} \Big)
   - \parder{C_2}{q^k}\Big( g^{-1}{}^{kl} \parder{P_1^i}{\dot{q}^l}
   - g^{-1}{}^{ij} \parder{P_1^k}{\dot{q}^j} \Big)
   \\&\qquad
    - (h_{kl}-h_{lk}) g^{-1}{}^{ij} \parder{}{\dot{q}^j}\Big( P_1^k P_2^l \Big)
   + c^{kj} g_{jl} 
   \Big( P_2^l   \parder{P_1^i}{\dot{q}^k} - P_1^l   \parder{P_2^i}{\dot{q}^k} \Big) . 
    \end{aligned}
\end{equation}
The first term with parenthesis in this expression
can be expanded and then simplified,
which yields 
\begin{equation}
  \begin{aligned}
    g^{-1}{}^{kl} \parder{P_2^i}{\dot{q}^l}
    - g^{-1}{}^{ij} \parder{P_2^k}{\dot{q}^j}
   & = g^{-1}{}^{kl} \parder{g^{-1}{}^{ij}}{\dot{q}^l} \parder{C_2}{\dot{q}^j} 
    - g^{-1}{}^{ij} \parder{g^{-1}{}^{kl}}{\dot{q}^j} \parder{C_2}{\dot{q}^l} \\
    & = - g^{-1}{}^{kl} g^{-1}{}^{im} \parder{g_{mj}}{\dot{q}^l} P_2^j
    + g^{-1}{}^{ij} g^{-1}{}^{km} \parder{g_{ml}}{\dot{q}^j} P_2^l \\
    & = g^{-1}{}^{ij} g^{-1}{}^{km} P_2^l
    \Big( \parder{g_{ml}}{\dot{q}^j} - \parder{g_{jl}}{\dot{q}^m} \Big)
  =0
  \end{aligned}
\end{equation}
due to the second commutativity relation in equation \eqref{commutativity.rels}.
Likewise, the second term with parenthesis in expression \eqref{P.lhs.rhs} vanishes,
\begin{equation}
g^{-1}{}^{kl} \parder{P_1^i}{\dot{q}^l}
- g^{-1}{}^{ij} \parder{P_1^k}{\dot{q}^j}
= 0 . 
\end{equation}
Expanding the third term in expression \eqref{P.lhs.rhs} yields
\begin{equation}
  \begin{aligned}
  -g^{-1}{}^{ij} (h_{kl}-h_{lk}) \parder{}{\dot{q}^j}\Big( P_1^k P_2^l \Big)
  & = -g^{-1}{}^{ij} (h_{kl}-h_{lk}) \Big(
  \parder{}{\dot{q}^j}\Big( g^{-1}{}^{km} \parder{C_1}{\dot{q}^m} \Big) P_2^l
  + \parder{}{\dot{q}^j}\Big( g^{-1}{}^{lm} \parder{C_2}{\dot{q}^m}\Big) P_1^k
  \Big) \\
    & = -g^{-1}{}^{ij} (h_{kl}-h_{lk})
  \Big( g^{-1}{}^{km} \parders{C_1}{\dot{q}^m}{\dot{q}^j} P_2^l
  + g^{-1}{}^{lm} \parders{C_2}{\dot{q}^m}{\dot{q}^j} P_1^k \\ &\qquad
  +\parder{g^{-1}{}^{km}}{\dot{q}^j} \parder{C_1}{\dot{q}^m} P_2^l
 + \parder{g^{-1}{}^{lm}}{\dot{q}^j} \parder{C_2}{\dot{q}^m} P_1^k
 \Big) . 
  \end{aligned}
\end{equation}
By means of integration by parts, the first two terms can be expressed as 
\begin{equation}
  \begin{aligned}
    - (h_{kl}-h_{lk}) \Big(
    g^{-1}{}^{km} \parder{P_1^i}{\dot{q}^m} P_2^l
  + g^{-1}{}^{lm} \parder{P_2^i}{\dot{q}^m} P_1^k 
  - g^{-1}{}^{km}   \parder{g^{-1}{}^{ij}}{\dot{q}^m} \parder{C_1}{\dot{q}^j} P_2^l
  - g^{-1}{}^{lm} \parder{g^{-1}{}^{ij}}{\dot{q}^m} \parder{C_2}{\dot{q}^j} P_1^k
  \Big) , 
  \end{aligned}
\end{equation}
which combines with the last two terms,
yielding
\begin{equation}\label{3rdterm}
  \begin{aligned}
    -g^{-1}{}^{ij} (h_{kl}-h_{lk}) \parder{}{\dot{q}^j}\Big( P_1^k P_2^l \Big)
    & = - (h_{kl}-h_{lk}) \bigg(
    g^{-1}{}^{km} \parder{P_1^i}{\dot{q}^m} P_2^l
    + g^{-1}{}^{lm} \parder{P_2^i}{\dot{q}^m} P_1^k \\&\qquad
    +\Big( g^{-1}{}^{ij} \parder{g^{-1}{}^{km}}{\dot{q}^j} \parder{C_1}{\dot{q}^m}
    - g^{-1}{}^{km} \parder{g^{-1}{}^{ij}}{\dot{q}^m} \parder{C_1}{\dot{q}^j}
    \Big) P_2^l \\&\qquad
    + \Big( g^{-1}{}^{ij} \parder{g^{-1}{}^{lm}}{\dot{q}^j} \parder{C_2}{\dot{q}^m}
    - g^{-1}{}^{lm} \parder{g^{-1}{}^{ij}}{\dot{q}^m} \parder{C_2}{\dot{q}^j}
    \Big) P_1^k
    \bigg) \\
    & = -c^{mk} g_{kl}   
    \Big( \parder{P_1^i}{\dot{q}^m} P_2^l - \parder{P_2^i}{\dot{q}^m} P_1^l \Big) \\&\qquad
    + c^{nk} g_{kl} g^{-1}{}^{ij} \parder{g_{nm}}{\dot{q}^j} \big( P_1^m P_2^l - P_2^m P_1^l \big)
   \\&\qquad\qquad
    -c^{mk} g_{kl} g^{-1}{}^{in} 
    \parder{g_{nj}}{\dot{q}^m} \big( P_1^j P_2^l - P_2^j P_1^l \big)
    \end{aligned}
\end{equation}
where the final equality is obtained using equation \eqref{h.matrix} along with antisymmetry of $h_{kl}-h_{lk}$. 
The first term on the right-hand side of expression \eqref{3rdterm}
cancels the fourth term in expression \eqref{P.lhs.rhs}.
This leaves only the terms in the last two lines in expression \eqref{3rdterm}. 
These terms can be combined after rearrangement of indices, yielding
\begin{equation}
  \begin{aligned}
  & c^{nk} g^{-1}{}^{ij} \Big( g_{kl}  \parder{g_{nm}}{\dot{q}^j} 
 - g_{km} \parder{g_{nl}}{\dot{q}^j}  \Big) P_1^m P_2^l
 -c^{mk}  g^{-1}{}^{in} \Big( g_{kl} \parder{g_{nj}}{\dot{q}^m}
 - g_{kj} \parder{g_{nl}}{\dot{q}^m} \Big)  P_1^j P_2^l
  \\&
 =   c^{nk} g^{-1}{}^{ij} \Big(
 g_{kl}  \parder{g_{nm}}{\dot{q}^j}  - g_{kl} \parder{g_{jm}}{\dot{q}^n}
 +  g_{km} \parder{g_{jl}}{\dot{q}^n} - g_{km} \parder{g_{nl}}{\dot{q}^j}
 \Big) P_1^m P_2^l
 =0
 \end{aligned}
\end{equation}
where the terms cancel in pairs 
due to the second commutativity relation in equation \eqref{commutativity.rels}.

Consequently, all terms have cancelled in the difference expression \eqref{P.lhs.rhs},
which establishes that $P^i_\text{l.h.s.} - P^i_\text{r.h.s.} =0$. 
This completes the proof.


\begin{thebibliography}{99}

\bibitem{WhiWat-book}
E.T. Whittaker, G.N. Watson,
{\it A Course of Modern Analysis},
Cambridge University Press, 1915.

\bibitem{Arn-book}
V.I. Arnold,
{\it Mathematical Methods of Classical Mechanics},
Graduate Texts in Mathematics Vol. 50 (2nd ed.), Springer, 1989.

\bibitem{BA-book}
G. Bluman and S.C. Anco,
{\it Symmetry and Integration Methods for Differential Equations},
Applied Math. Sci. Volume 154
(Springer, New York) 2002.

\bibitem{Mil.Pos.Win}
W. Miller, S. Post, P. Winternitz,
Classical and quantum superintegrability with applications,
J. Phys. A 46 (2013), 423001.

\bibitem{Fra}
D.M. Fradkin,
 Existence of the dynamic symmetries $O_4$ and $SU(3)$ for all classical central potential problems, 
Prog. Theor. Phys. 37 (1967), 798--812.

\bibitem{Per}
A. Peres,
Generalisation of the Runge-Lenz constant of classical motion in a central potential,
J. Phys. A: Math. Gen. 12 (1979), 1711--1713.

\bibitem{GolPooSaf}
H. Goldstein, C. Poole, J. Safko, 
{\it Classical Mechanics} (3rd ed.), (Addison Wesley) 2000.

\bibitem{AncMeaPas}
S.C. Anco, T. Meadows, V. Pascuzzi, 
Some new aspects of first integrals and symmetries for central force dynamics, 
J. Math. Phys. 57 (2016), 062901 (35 pages).

\bibitem{AncBalGan}
S.C. Anco, A. Ballesteros, M. Gandarias, 
Global versus local (super)integrability of a nonlinear oscillator, 
Phys. Lett. A. 383 (2019), 801--807. 

\bibitem{Olv-book}
P.J. Olver,
{\it Applications of Lie Groups to Differential Equations},
(Springer, New York) 1986.

\bibitem{BCA-book}
G. Bluman, A.F. Cheviakov, and S.C. Anco, 
{\it Applications of Symmetry Methods to Partial Differential Equations}, 
Applied Mathematical Sciences series, Volume 168, Springer (2009). 

\bibitem{Noe}
E. Noether,
Invariante Variationsprobleme,
Nachrichten von der Gesellschaft der Wissenschaften zu G\"ottingen. Mathematisch-Physikalische Klasse (1918), 235-–257.

\bibitem{Anc-review}
S.C. Anco, 
Generalization of Noether's theorem in modern form to non-variational partial differential equations. 
In: Recent progress and Modern Challenges in Applied Mathematics, Modeling and Computational Science, 119--182, 
Fields Institute Communications, Volume 79, Springer (2017).

\bibitem{AncBao}
S.C. Anco, W. Bao,
A formula for symmetry recursion operators from non-variational symmetries of partial differential equations, 
Lett. Math. Phys. 111 (2021), 70 (33 pages).

\bibitem{Bou.Bou}
S. Bouquet and A. Bourdier,
Notion of integrability for time-dependent Hamiltonian systems:
Illustrations from the relativistic motion of a charged particle,
Phys. Rev. E 57(2) (1998), 1273--1283.

\bibitem{Ley}
F. Leyvraz, 
Liouville integrability may not be what you think, 
Am. J. Phys. 93 2025), 320--327.

\bibitem{Anc2026}
S.C. Anco,
Noether symmetry groups, locally conserved integrals, and dynamical symmetries in classical mechanics.
arXiv: 2603.26624 







\end{thebibliography}
\end{document}